\documentclass[journal]{IEEEtran}

\usepackage{epsfig,graphics}
\usepackage{amsmath,amssymb,amsthm,multirow}
\usepackage{cite}
\usepackage{enumerate}
\usepackage{subfigure}
\usepackage{tikz,pgfplots}
\newtheorem{theorem}{Theorem}

\newtheorem{assumption}{Assumption}

\usepackage{color}  

\newcommand{\cb}[1]{{\boldsymbol{#1}}}
\newcommand{\cp}[1]{\ifmmode {\mathcal{#1}}\else ${\mathcal{#1}}$\fi}

\newcommand{\balpha}{\boldsymbol{\alpha}}
\newcommand{\bbeta}{\boldsymbol{\beta}}

\newcommand{\bGamma}{\boldsymbol{\Gamma}}
\newcommand{\bdelta}{\boldsymbol{\delta}}

\newcommand{\bpsi}{\boldsymbol{\psi}}
\newcommand{\bphi}{\boldsymbol{\phi}}
\newcommand{\bPhi}{\boldsymbol{\Phi}}

\newcommand{\bphit}{\widetilde\bphi}

\newcommand{\bp}{\boldsymbol{p}}

\newcommand{\bg}{\boldsymbol{g}}
\newcommand{\bw}{\boldsymbol{w}}
\newcommand{\bx}{\boldsymbol{x}}

\newcommand{\bv}{\boldsymbol{v}}
\newcommand{\bz}{\boldsymbol{z}}

\newcommand{\bwt}{\widetilde\bw}

\newcommand{\bC}{\boldsymbol{C}}

\newcommand{\bG}{\boldsymbol{G}}

\newcommand{\br}{\boldsymbol{r}}
\newcommand{\bW}{\boldsymbol{W}}

\newcommand{\cAT}{\boldsymbol{\cal{A}}^{\top}}
\newcommand{\cA}{\boldsymbol{\cal{A}}}
\newcommand{\cX}{\boldsymbol{\cal{X}}}

\newcommand{\cC}{\boldsymbol{\cal{C}}}
\newcommand{\cCT}{\boldsymbol{\cal{C}}^\top}

\newcommand{\cM}{\boldsymbol{\cal{M}}}

\newcommand{\cB}{\boldsymbol{\cal{B}}}
\newcommand{\cR}{\boldsymbol{\cal{R}}}
\newcommand{\cF}{\boldsymbol{\cal{F}}}

\newcommand{\bA}{\boldsymbol{A}}
\newcommand{\bR}{\boldsymbol{R}}
\newcommand{\bX}{\boldsymbol{X}}

\newcommand{\bU}{\boldsymbol{U}}

\newcommand{\bSig}{\boldsymbol{\Sigma}}
\newcommand{\bsig}{\boldsymbol{\sigma}}

\newcommand{\bUps}{\boldsymbol{\Upsilon}}
\newcommand{\bI}{\boldsymbol{I}}

\newcommand{\cI}{\mathcal{I}}
\newcommand{\cO}{\mathcal{O}}

\newcommand{\expec}{\mathbb{E}}

\newcommand{\N}{{\cp{N}}}

\newcommand{\C}{{\cp{C}}}

\newcommand{\tr}{\text{Tr}}
\newcommand{\vc}{\text{vec}}
\newcommand{\col}{\text{col}}

\newcommand{\diag}{\text{diag}}
\newcommand{\prox}{\text{prox}}
\newcommand{\card}{\text{card}}

\DeclareMathOperator*{\argmin}{argmin}

\usetikzlibrary{snakes}

\begin{document}
\title{Proximal Multitask Learning over Networks with Sparsity-inducing Coregularization}

\author{Roula Nassif, C{\'e}dric Richard, \IEEEmembership{Senior Member, IEEE} \\
Andr{\'e} Ferrari, \IEEEmembership{Member, IEEE}, Ali H. Sayed, \IEEEmembership{Fellow Member, IEEE}
\thanks{The work of C. Richard and A. Ferrari was partly supported by ANR and DGA grant ANR-13-ASTR-0030 (ODISSEE project). The work of A. H. Sayed was supported in part by NSF grants CIF-1524250 and ECCS-1407712.
R. Nassif, C. Richard, and A. Ferrari are with the Universit\'e C\^ote d'Azur, OCA, CNRS, France (email: roula.nassif@oca.eu; cedric.richard@unice.fr; andre.ferrari@unice.fr). C. Richard is on leave at INRIA Sophia Antipolis - M\'editerran\'ee, France.

A. H. Sayed is with the department of electrical engineering, University of California, Los Angeles, USA (email: sayed@ee.ucla.edu).
}
}

\maketitle
\begin{abstract}
In this work, we consider multitask learning problems where clusters of nodes are interested in estimating their own parameter vector. Cooperation among clusters is beneficial when the optimal models of adjacent clusters have a good number of similar entries. We propose a fully distributed algorithm for solving this problem. The approach relies on minimizing a global mean-square error criterion regularized by non-differentiable terms to promote cooperation among neighboring clusters. A general diffusion forward-backward splitting strategy is introduced. Then, it is specialized to the case of sparsity promoting regularizers. A closed-form expression for the proximal operator of a weighted sum of $\ell_1$-norms is derived to achieve higher efficiency. We also provide conditions on the step-sizes that ensure convergence of the algorithm in the mean and mean-square error sense. Simulations are conducted to illustrate the effectiveness of the strategy.
\end{abstract}


\maketitle

\section{Introduction}
We consider the problem of distributed adaptive learning over networks to simultaneously estimate several parameter vectors from noisy measurements using in-network processing. Depending on the number of parameter vectors to estimate, we distinguish between single-task networks and multitask networks. In a single-task scenario, the entire network aims to estimate a common parameter vector for all nodes. The nodes are allowed to exchange information with their neighbors to improve their own estimates. Then, the estimates are combined in order to achieve the solution of the problem. Different cooperation rules have been proposed and studied in the literature~\cite{Bertsekas1997,Nedic2001,Lopes2007incr,sayed2014adaptivenet,sayed2014adaptation,Sayed2013diff,sayed2014diffusion,lopes2008diffusion,Cativelli2010diffusion,ChenUCLA2012,chen2013pareto,dimakis2010gossip,nedic2009distributed,kar2011convergence,lee2013distributed,gharehshiran2013energyaware,ram2010distributed}. Diffusion strategies~\cite{sayed2014adaptivenet,sayed2014adaptation,Sayed2013diff,sayed2014diffusion,lopes2008diffusion,Cativelli2010diffusion,ChenUCLA2012,chen2013pareto} are particularly attractive since they are scalable, robust, and enable continuous learning and adaptation in response to concept drifts. They have also been shown to outperform consensus implementations over adaptive networks when constant step-sizes are employed to enable continuous adaptation~\cite{tu2012diffusion,sayed2014adaptivenet,sayed2014adaptation}.

In this work, we are interested in distributed estimation over multitask networks: nodes are grouped into clusters, and each cluster is interested in estimating its own parameter vector (i.e., each cluster has its own task). Although clusters may generally have distinct though related tasks to perform, the nodes may still be able to capitalize on inductive transfer between clusters to improve their estimation accuracy. Such situations occur when the tasks of nearby clusters are correlated, which happens, for instance, in monitoring applications where agents in a network need to track multiple targets moving along correlated trajectories. Multitask diffusion estimation problems of this type have been addressed before in two main ways. 

In a first scenario, no prior information on possible relationships between tasks is assumed and nodes do not know which other nodes share the same task. In this case, all nodes cooperate with each other as dictated by the network topology. It was shown in~\cite{chen2013pareto} that the diffusion iterates will end up converging to a Pareto optimal solution corresponding to a multi-objective optimization problem. If, on the other hand, the only available information is that clusters may exist in the network (but their structures are not known), then extended diffusion strategies can be developed~\cite{zhao2013clustering,Chen2014diffusionovermultitask,zhao2014distributedclustering,khawatmi2015decentralized} for setting the combination weights in an online manner in order to enable automatic network clustering and, subsequently, to limit cooperation between clustered agents. In a second scenario, it is assumed that nodes know which clusters they belong to. In this case, multitask diffusion strategies can be derived by exploiting this information on the relationships between tasks. A couple of useful works have addressed variations of this scenario. For example, in~\cite{abdolee2014estimation}, a diffusion LMS strategy estimates spatially-varying parameters by exploiting the spatio-temporal correlations of the measurements at neighboring nodes. In~\cite{plata2014distributed}, it is assumed that there are three types of parameters: parameters of global interest to all nodes in the network, parameters of common interest to a subset of nodes, and a collection of parameters of local interest. A diffusion strategy was developed to perform estimation under these conditions. {A similar work dealing with incremental strategies instead of diffusion strategies appears in~\cite{bodanovic2014distributed}. Likewise, in the works~\cite{bertrand2010NSPE1,bertrand2011NSPE2}, distributed algorithms are developed to estimate node-specific parameter vectors that lie in a common latent signal subspace. In another work~\cite{chen2014overlappinghypothesis}, the parameter space is decomposed into two orthogonal subspaces, with one of the subspaces being common to all nodes. There is yet another useful way to exploit and model relationships among tasks, namely, to formulate optimization problems with appropriate co-regularizers between nodes. The strategy developed in~\cite{chen2014multitask} adds squared $\ell_2$-norm co-regularizers to the mean-square-error criterion in order to promote smoothness of the graph signal. Its convergence behavior is studied over asynchronous networks in~\cite{nassif2014asynmulti}.

In some applications, however, such as cognitive radio~\cite{chen2014overlappinghypothesis,plata2014distributed} and remote sensing~\cite{chen2014multitask}, it may happen that the optimum parameter vectors of neighboring clusters have a large number of similar entries and a relatively small number of distinct components. In this work, we build on the second scenario where the composition of the clusters is assumed to be known and where nodes know which cluster they belong to. It is then advantageous to develop distributed strategies that involve cooperation among adjacent clusters in order to promote and exploit such similarity. Although the current problem seems to be related to the problem studied in~\cite{chen2014multitask}, it should be noted that the differentiable regularizers used in~\cite{chen2014multitask} are not effective when sparsity promoting regularization is required. Moreover, when neighboring nodes belonging to different clusters are aware of the indices of common and distinct entries, and when these indices are fixed over time, one may appeal to the multitask diffusion strategies developed in~\cite{plata2014distributed,chen2014overlappinghypothesis}. However, in the current work, we are interested in solutions that are able to handle situations where the only available information is that the optimum parameter vectors of neighboring clusters have a large number of similar entries. A multitask diffusion algorithm with $\ell_1$-norm co-regularizers is proposed in~\cite{nassif2015icassp} to address this problem leading to a subgradient descent method distributed among the agents. The aim of this work is to introduce a more general approach for solving such convex but \emph{non-differentiable} problems {by employing instead a diffusion forward-backward splitting strategy based on the proximal projection operator. Before proceeding, we recall the forward-backward splitting approach in a single-agent deterministic environment~\cite{combettes2011proximal,combettes2011proximity,parikh2013proximal}. 

Consider the problem 
\begin{equation}
	\label{eq: deterministic problem}
	\min\limits_{\bx\in\mathbb{R}^M}f(\bx)+ g(\bx)
\end{equation}
with $f$ a real-valued differentiable convex function whose gradient is $\beta$-Lipschitz continuous, and $g$ a real-valued convex function. The \textit{proximal gradient} method or the \textit{forward-backward splitting} approach for solving~\eqref{eq: deterministic problem} is given by the iteration~\cite{combettes2011proximal,parikh2013proximal}:
\begin{equation}
	\label{eq: FBS deterministic}
	\bx(i+1)=\text{prox}_{\mu g}(\bx(i)-\mu\nabla f(\bx(i))),
\end{equation}
where $\mu$ is a constant step-size chosen such that $\mu\in(0,2\beta^{-1}]$ to ensure convergence to the minimizer of~\eqref{eq: deterministic problem}. The gradient-descent step is the forward step (explicit step) and the proximal step is the backward step (implicit step). The proximal operator of $\mu g(\bx)$ at a given point $\bv\in\mathbb{R}^M$ is a real-valued map given by~\cite{parikh2013proximal}:
\begin{equation}
	\label{eq: proximal operator}
	\text{prox}_{\mu g}(\bv)=\argmin\limits_{\bx\in\mathbb{R}^M}~g(\bx)+\frac{1}{2\mu}\|\bx-\bv\|^2.
\end{equation}
Since the proximal operator needs to be calculated at each iteration in~\eqref{eq: FBS deterministic}, it is important to have a closed form expression for evaluating it. In this work, we derive a \textit{multitask diffusion adaptation strategy} where each node employs this approach for minimizing a cost function with sparsity based co-regularizers. Instead of using iterative algorithms for evaluating the proximal operator of a weighted sum of $\ell_1$-norms at each iteration~\cite{combettes2011proximity}, we shall instead derive a closed form expression that allows us to compute it exactly. We shall also examine under which conditions on the step-sizes the proposed multitask diffusion strategy is mean and mean-square stable. Simulations are conducted to show the effectiveness of the proposed strategy. An adaptive rule to guarantee an appropriate cooperation between clusters is also introduced.

\textbf{Notation.} In what follows, normal font letters denote scalars, boldface lowercase letters denote column vectors, and boldface uppercase letters denote matrices. We use the symbol $(\cdot)^\top$ to denote matrix transpose, the symbol $(\cdot)^{-1}$ to denote matrix inverse, and the symbol $\tr(\cdot)$ to denote the trace operator. The operator $\col\{\cdot\}$ stacks the column vectors entries on top of each other. The symbol $\otimes$ denotes the Kronecker product operation. The identity matrix of size $N\times N$ is denoted by $\bI_N$. The $N\times M$ matrices of zeros and ones are denoted by $\cb{0}_{N\times M}$ and $\cb{1}_{N\times M}$, respectively. The set $\N_k$ denotes the neighbors of node $k$ including $k$. The set $\N_k^-$ denotes the neighbors of node $k$ excluding $k$. Finally, $\C_i$ denotes the set of nodes in the $i$-th cluster and $\C(k)$ denotes the cluster to which node $k$ belongs.


\section{Multitask diffusion LMS with Forward-Backward splitting}
\label{sec:Multitask Diffusion LMS}

\subsection{Network model and problem formulation}
\label{Problem formulation}

We consider a network of $N$ nodes grouped into $Q$ connected clusters in a predefined topology. Clusters are assumed to be connected, i.e., there exists a path between any pair of nodes in the cluster. At every time instant $i$, every node $k$ has access to a zero-mean measurement $d_k(i)$ and a zero-mean $M\times 1$ regression vector $\bx_k(i)$ with positive covariance matrix $\bR_{\bx,k}=\expec\{\bx_k(i)\,\bx_k^\top(i)\}>0$. We assume the data to be related via the linear model: 
\begin{equation}
	\label{eq: linear data model}
	d_k(i)=\bx_k^\top(i)\bw^o_k+z_k(i),
\end{equation}
where $\bw^o_k$ is the $M\times 1$ unknown parameter vector, also called task, we wish to estimate at node $k$, and $z_k(i)$ is a zero-mean measurement noise of variance $\sigma_{z,k}^2$, independent of  $\bx_{\ell}(j)$ for all $\ell$ and $j$, and independent of $z_{\ell}(j)$ for $\ell\neq k$ or $i\neq j$. We assume that all nodes in a cluster are interested in estimating the same parameter vector, namely, $\bw^o_k=\bw^o_{\C_q}$ whenever $k$ belongs to cluster $\C_q$. However, if cluster $\C_p$ is connected to cluster $\C_q$, that is, there exists at least one link connecting a node from $\C_p$ to a node from $\C_q$, vectors $\bw^o_{\C_p}$ and $\bw^o_{\C_q}$ are assumed to have a large number of similar entries and only a relatively small number of distinct entries. Cooperation across these clusters can therefore be beneficial to infer $\bw^o_{\C_p}$ and $\bw^o_{\C_q}$. 

Considerable interest has been shown in the literature about estimating an optimum parameter vector $\bw^o$ subject to the property of being sparse. Motivated by the well-known LASSO problem~\cite{tibshirani1996regression} and compressed sensing framework~\cite{baraniuk2007cs}, different techniques for sparse adaptation have been proposed. For example, the authors in~\cite{Chen2009sparselms,Gu2009l0norm} promote sparsity within an LMS framework by considering regularizers based on the $\ell_1$-norm, reweighed $\ell_1$-norm, and convex approximation of $\ell_0$-norm. In~\cite{kopsinis2011projection}, projections of streaming data onto hyperslabs and weighted $\ell_1$ balls are used instead of minimizing regularized costs recursively. Proximal forward-backward splitting is considered in an adaptive scenario in~\cite{mukrami2010adaptivprox}. In the context of distributed learning over \emph{single-task} networks, diffusion LMS methods promoting sparsity have been proposed. Sparse diffusion LMS strategies using subgradient methods are proposed in~\cite{lorenzo2012ICASSP,Lorenzo2013sparse,Liu2012sparsediffusion} and using proximal methods are proposed in~\cite{wee2013proxdiffusion,lorenzo2014GMRF,vlaski2015prox}. In~\cite{chouvardas2012sparsity}, the authors employ projection-based techniques~\cite{kopsinis2011projection} to derive distributed diffusion algorithms promoting sparsity, and in~\cite{chouvardas2013greedy} a diffusion LMS algorithm for estimating an $s$-sparse vector is proposed based on adaptive greedy techniques similar to~\cite{mileounis2010greedy}. These techniques estimate the positions of non-zero entries in the target vector, and then perform computations on this subset. More generally, diffusion strategies based on proximal gradient for minimizing general costs (not necessarily mean-square error costs) and subject to a broader class of constraints on the parameter vector to be estimated (including sparsity) are derived in~\cite{vlaski2015prox}. 

Our purpose is to derive an adaptive learning algorithm over \emph{multitask} networks where optimum parameter vectors of neighboring clusters share a large number of similar entries and a relatively small number of distinct entries. Consider nodes $k$ and $\ell$ of neighboring clusters $\C(k)$ and $\C(\ell)$, and let $\bdelta_{k,\ell}$ denote the vector difference $\bw_{\C(k)}-\bw_{\C(\ell)}$. Promoting the sparsity of $\bdelta_{k,\ell}$ can be performed by considering the pseudo $\ell_0$-norm of $\bdelta_{k,\ell}$ as it denotes the number of nonzero entries. Nevertheless, $\|\bdelta_{k,\ell}\|_0$ is a non-convex co-regularizer that leads to computational challenges}. A common alternative is to use the $\ell_1$-norm regularization function defined as
\begin{equation}
	\label{eq: l1 norm}
	f_1(\bdelta_{k,\ell})=\|\bdelta_{k,\ell}\|_1=\sum_{m=1}^M |[\bdelta_{k,\ell}]_m|.
\end{equation}
Since the $\ell_1$-norm uniformly shrinks all the components of a vector and does not distinguish between zero and non-zero entries~\cite{candes2008enhancing}, it is common in the sparse adaptive filtering framework~\cite{Chen2009sparselms,kopsinis2011projection,mukrami2010adaptivprox,lorenzo2012ICASSP,Lorenzo2013sparse,wee2013proxdiffusion,lorenzo2014GMRF,chouvardas2012sparsity,gao2013forwardbackward} to consider a weighted formulation of the $\ell_1$-norm. Weighted  $\ell_1$-norm was designed to reduce the bias induced by the $\ell_1$-norm and enhance the penalization of the non-zero entries of a vector~\cite{candes2008enhancing,kopsinis2011projection,zou2006adaptive}. Given the weight vector $\balpha_{k\ell}=[\alpha_{k\ell}^1,\ldots,\alpha_{k\ell}^M]^\top$, with $\alpha_{k\ell}^m > 0 $ for all $m$, the weighted $\ell_1$-norm is defined as:
\begin{equation}
	\label{eq: reweighted l1 norm}
	f_2(\bdelta_{k,\ell})=\sum_{m=1}^M\alpha_{k\ell}^m\big|[\bdelta_{k,\ell}]_m\big|.
\end{equation}
The weights are usually chosen as:
\begin{equation}
\alpha_{k\ell}^m=\frac{1}{\epsilon+|[\bdelta_{k,\ell}^o]_m|}, \quad m=1,\ldots,M,
\end{equation}
where $\bdelta_{k,\ell}^o=\bw_k^o-\bw_\ell^o$. Since the optimum parameter vectors are not available beforehand, we set
\begin{equation}
\label{eq: weights}
	\alpha_{k\ell}^m(i)=\frac{1}{\epsilon+|[\bdelta_{k,\ell}(i-1)]_m|}, \quad m=1,\ldots,M,
\end{equation}
at each iteration $i$, where $\epsilon$ is a small constant to prevent the denominator from vanishing and $\bdelta_{k,\ell}(i)$ is the estimate of $\bdelta_{k,\ell}^o$ at nodes $k$ and $\ell$ and iteration $i$. This technique, also known as reweighted $\ell_1$ minimization~\cite{candes2008enhancing}, is performed at each iteration of the stochastic optimization process. It has been shown in~\cite{candes2008enhancing} that, by minimizing~\eqref{eq: reweighted l1 norm} with the weights~\eqref{eq: weights}, one minimizes the log-sum penalty function, $\sum_{m=1}^M\log(\epsilon+|[\bdelta_{k,\ell}]_m|)$, which acts like the $\ell_0$-norm by allowing a relatively large penalty to be placed on small nonzero coefficients and more strongly encourages them to be set to zero. In the sequel, we shall use $f(\bw_{\C(k)}-\bw_{\C(\ell)})$ to refer to the unweighted or reweighted $\ell_1$-norm promoting the sparsity of $\bw_{\C(k)}-\bw_{\C(\ell)}$. 

It is sufficient for this work to derive a distributed learning algorithm of the LMS type. We shall therefore assume that the local cost function $J_k(\bw_{\C(k)})$ at node $k$ is the mean-square error criterion defined by:
\begin{equation}
	\label{eq: local cost function}
	J_k(\bw_{\C(k)})=\expec\big\{| d_k(i)-\bx_k^\top(i)\bw_{\C(k)}|^2\big\}.
\end{equation}
Combining local mean-square-error cost functions and regularization functions, the {cooperative multitask estimation problem} is formulated as the problem of seeking a fully distributed solution for solving:
\begin{equation}
\begin{split}
	\label{eq: problem P}
	\hspace{-2mm}\min_{\bw_{\C_1},\ldots,\bw_{\C_Q}}\hspace{-2mm}\overline{J}^\text{glob}(\bw_{\C_1}&,\ldots,\bw_{\C_Q})=\min_{\bw_{\C_1},\ldots,\bw_{\C_Q}}\sum_{k=1}^NJ_k(\bw_{\C(k)})\\
	&+\eta\sum_{k=1}^N\hspace{-0.5mm}\sum_{\ell\in \N_k\setminus\C(k)}\hspace{-3.5mm}\rho_{k\ell}f(\bw_{\C(k)}-\bw_{\C(\ell)}),
\end{split}
\end{equation}
where $\eta>0$ is the regularization strength used to enforce sparsity. It ensures a tradeoff between fidelity to the measurements and prior information on the relationships between tasks. The weights $\rho_{k\ell}\geq0$ aim at locally adjusting the regularization strength. The notation $\N_k\setminus\C(k)$ denotes the set of neighboring nodes of $k$ that are not in the same cluster as $k$.

Note that the regularization terms \eqref{eq: l1 norm} and \eqref{eq: reweighted l1 norm} are symmetric with respect to the weight vectors $\bw_{\C(k)}$ and $\bw_{\C(\ell)}$, that is, $f(\bw_{\C(k)}-\bw_{\C(\ell)})=f(\bw_{\C(\ell)}-\bw_{\C(k)})$. Due to the summation over the $N$ nodes, each term $f(\bw_{\C(k)}-\bw_{\C(\ell)})$ can be viewed as weighted by $\frac{(\rho_{k\ell}+\rho_{\ell k})}{2}$ in~\eqref{eq: problem P}. Problem~\eqref{eq: problem P}  can therefore be written in an alternative way as:
\begin{equation}
\label{eq: problem P'}
\begin{split}
	\min_{\bw_{\C_1},\ldots,\bw_{\C_Q}}\hspace{-2mm}\overline{J}^\text{glob}(\bw_{\C_1},\ldots,&\bw_{\C_Q})
	=\min_{\bw_{\C_1},\ldots,\bw_{\C_Q}}\sum_{k=1}^NJ_k(\bw_{\C(k)})\\
	&+\eta\sum_{k=1}^{N}\hspace{-0.5mm} \sum_{\ell\in \N_k\setminus\C(k)}\hspace{-3.5mm}p_{k\ell}f(\bw_{\C(k)}-\bw_{\C(\ell)})
\end{split}
\end{equation}
where the factors $\{p_{k\ell}\}$ are symmetric, i.e., $p_{k\ell}=p_{\ell k}$, and are given by:
\begin{equation}
p_{k\ell}\triangleq\frac{(\rho_{k\ell}+\rho_{\ell k})}{2}.
\end{equation}
One way to avoid symmetrical regularization is to consider an alternative problem formulation defined in terms of $Q$ Nash equilibrium problems as done in~\cite{chen2014multitask} with $\ell_2$-norm co-regularizers. In this paper, we shall focus on problem~\eqref{eq: problem P}.

Let us consider the variable $\bw_{\C_j}$ of the $j$-th cluster. Given $\bw_{\C(\ell)}$ with $\ell\in\N_k\setminus\C_j$ and $k\in\C_j$, the subdifferential of $\overline{J}^\text{glob}(\bw_{\C_1},\ldots,\bw_{\C_Q})$ in~\eqref{eq: problem P'} with respect to $\bw_{\C_j}$ is given by: 
\begin{align}
	&\partial_{\bw_{\C_j}}\overline{J}^\text{glob}(\bw_{\C_1},\ldots,\bw_{\C_Q})\nonumber\\
	&=\sum_{k\in\C_j}\hspace{-1mm}\nabla_{\bw_{\C_j}}\hspace{-1mm}J_k(\bw_{\C_j})+2\eta\hspace{-0.5mm}\sum_{k\in\C_j}\hspace{-.5mm} \sum_{\ell\in \N_k\setminus\C_j}\hspace{-3.25mm}p_{k\ell}\partial_{\bw_{\C_j}}\hspace{-0.5mm}f(\bw_{\C_j}-\bw_{\C(\ell)}),
\end{align}
where we used the fact that the regularization terms~\eqref{eq: l1 norm}, \eqref{eq: reweighted l1 norm}, and the regularization factors $\{p_{k\ell}\}$  are symmetric. Since we are interested in a distributed strategy for solving~\eqref{eq: problem P} that relies only on in-network processing, we associate the following regularized problem $(\cp{P}_j)$ with each cluster $\cp{C}_j$:
\begin{equation}
\begin{split}
	\label{eq: P_j}
	\min_{\bw_{\C_j}} \overline{J}_{\C_j}(\bw_{\C_j})=\min_{\bw_{\C_j}}&
	\sum_{k\in\C_j}\expec\big\{|d_k(i)-\bx_k^\top(i)\bw_{\C_j}|^2\big\}+\\
	&2\eta\sum_{k\in{\C_j}}\sum_{\ell\in \N_k\setminus\C_j}p_{k\ell}f(\bw_{\C_j}-\bw_{\C(\ell)}).
\end{split}
\end{equation}
Given $\bw_{\C(\ell)}$ with $\ell\in\N_k\setminus\C_j$, note that the costs in problems~\eqref{eq: problem P} and \eqref{eq: P_j} have the same subdifferential relative to $\bw_{\C_j}$.  In order that each node can solve the problem in an autonomous and adaptive manner using only local interactions, we shall derive a distributed iterative algorithm for solving~\eqref{eq: problem P} by considering \eqref{eq: P_j} since both costs have the same subdifferential information.

\subsection{Problem relaxation}

We shall now extend the derivations in~\cite{sayed2014diffusion,Cativelli2010diffusion,Chen2014ICASSP} to handle multitask estimation problems with nondifferentiable functions. In the sequel, we write $\bw_k$ instead of $\bw_{\C(k)}$ for simplicity of notation. First, we associate with each node $k$ an unregularized local cost function $J_k^{\text{loc}}(\cdot)$ and a regularized local cost function $\overline{J}_k^{\text{loc}}(\cdot)$ of the form:
\begin{equation}
	\label{eq: unregularized neighborhood function}
	J_k^{\text{loc}}(\bw_k)=\sum_{\ell\in\N_k\cap\C(k)}c_{\ell k}
	\expec\big\{|d_\ell(i)-\bx_{\ell}^\top(i)\bw_k |^2\big\},
\end{equation}
\begin{equation}
	\label{eq: regularized neighborhood function}
	\begin{split}
	\overline{J}_k^{\text{loc}}(\bw_{k})=&\sum\limits_{\ell\in\N_k\cap\C(k)}c_{\ell k}\expec\big\{|d_\ell(i)-\bx_{\ell}^\top(i)\bw_k |^2\big\}+\\
	&\quad2\eta\sum_{\ell\in\N_k\setminus\C(k)}p_{k\ell}f(\bw_k-\bw_\ell),
	\end{split}
\end{equation}
where $\N_k\cap\C(k)$ denotes the set of nodes in the neighborhood of node $k$ that belongs to its cluster, and $\{c_{\ell k}\}$ are non-negative weights satisfying 
\begin{equation}
	\label{eq: conditions on c}
	\sum_{k=1}^N c_{\ell k}=1,\quad\text{and}\quad c_{\ell k}=0\quad \text{if}\quad k\notin\N_\ell\cap\C(\ell).
\end{equation}
Note that $\bw_k=\bw_\ell$ whenever $\ell\in\N_k\cap\C(k)$. Both costs \eqref{eq: unregularized neighborhood function} and \eqref{eq: regularized neighborhood function} consist of a convex combination of mean-square errors in the neighborhood of node $k$ but limited to its cluster. In addition, expression \eqref{eq: regularized neighborhood function} takes interactions among neighboring clusters into account. Let us consider node $k$ belonging to cluster $\C_j$, i.e., $\C_j=\C(k)$. It can be checked that $\overline{J}_{\C_j}(\bw_{\C_j})$ in $\eqref{eq: P_j}$ can be written as:
\begin{equation}
	\label{eq: P_j local costs}
	\overline{J}_{\C_j}(\bw_{\C_j})
	=\sum_{\ell\in\C_j}\overline{J}_\ell^{\text{loc}}(\bw_{\ell})
	=\overline{J}_k^{\text{loc}}(\bw_{k})+\sum_{\ell\in\C_j\setminus\{k\}}\overline{J}_\ell^{\text{loc}}(\bw_{\ell}),
\end{equation}
The term $\sum_{\ell\in\C_j\setminus\{k\}}\overline{J}_\ell^{\text{loc}}(\bw_{\ell})$ contains terms promoting relationships between nodes $\ell\in\C_j\setminus\{k\}$ and their neighbors that are outside $\C_j$ but not necessarily in the neighborhood of node $k$. To limit these inter-cluster information exchanges to node $k$ and its extra-cluster neighbors, we relax $\sum_{\ell\in\C_j\setminus\{k\}}\overline{J}_\ell^{\text{loc}}(\bw_{\ell})$ to $\sum_{\ell\in\C_j\setminus\{k\}}J_\ell^{\text{loc}}(\bw_{\ell})$. 
Since~\eqref{eq: unregularized neighborhood function} is second-order differentiable, a completion-of-squares argument shows that each $J_\ell^{\text{loc}}(\bw_{\ell})$ can be expressed as~\cite{sayed2014diffusion}:
\begin{equation}
	\label{eq: unregularized neighborhood function 2}
	J_\ell^{\text{loc}}(\bw_{\ell})
	=J_\ell^{\text{loc}}(\bw_{\ell}^{\text{loc}})+\|\bw_\ell-\bw_\ell^{\text{loc}}\|^2_{\bR_\ell},
\end{equation}
where the notation $\|\bx\|^2_{\bSig}$ denotes $\bx^\top\bSig\,\bx$  for any nonnegative definite matrix $\bSig$, $\bw^{\text{loc}}_\ell$ is the minimizer of $J_\ell^{\text{loc}}(\bw_{\ell})$, and $\bR_\ell$ is given by:
\begin{equation}
	\bR_\ell=\sum_{k\in\N_\ell\cap\C(\ell)}c_{k \ell}	\bR_{\bx,k}.
\end{equation}
Thus, using \eqref{eq: regularized neighborhood function}, \eqref{eq: P_j local costs}, and \eqref{eq: unregularized neighborhood function 2} and dropping the constant term $J_\ell^{\text{loc}}(\bw_{\ell}^{\text{loc}})$, we can replace the original cluster cost~\eqref{eq: P_j} by the following cost function for cluster $\C(k)$ at node $k$:
\begin{equation}
	\label{eq: cluster function 1 relative to node k}
	\begin{split}
	\overline{J}'_{\C(k)}(\bw_k&)
	=\sum_{\ell\in\N_k\cap\C(k)}c_{\ell k}\expec\big\{|d_\ell(i)-\bx_{\ell}^\top(i)\bw_k |^2\big\}+\\
	&2\eta\hspace{-2mm}\sum_{\ell\in\N_k\setminus\C(k)}\hspace{-0.35cm}p_{k\ell}f(\bw_k-\bw_\ell)
	+\hspace{-2mm}\sum_{\ell\in\C(k)\setminus\{k\}}\hspace{-0.35cm}\|\bw_\ell-\bw_\ell^{\text{loc}}\|^2_{\bR_\ell}.
	\end{split}
\end{equation}
Equation \eqref{eq: cluster function 1 relative to node k} is an approximation relating the local cost function $\overline{J}_k^{\text{loc}}(\bw_{k})$ at node $k$ to the global cost function~\eqref{eq: P_j} associated with the cluster $\C(k)$. Node $k$ cannot minimize~\eqref{eq: cluster function 1 relative to node k} directly since this cost still requires global information that may not be available in its neighborhood. To avoid access to information via multihop, we relax $\overline{J}'_{\C(k)}(\bw_k)$ by limiting the sum in the third term on the RHS of \eqref{eq: cluster function 1 relative to node k} over the neighbors of node $k$. In addition, since the covariance matrices $\bR_{\bx,\ell}$ may not be known beforehand within the context of online learning, a useful strategy proposed in~\cite{sayed2014diffusion} is to substitute the covariance matrices $\bR_\ell$ by diagonal matrices of the form $b_{\ell k}\bI_M$, where $b_{\ell k}$ are nonnegative coefficients that allow to assign different weights to different neighbors. Later, these coefficients will be incorporated into a left stochastic matrix and the designer does not need to worry about their selection. Based on the arguments presented so far, the cluster cost function at each node $k$ can be relaxed as follows: 
\begin{equation}
	\label{eq: cluster function 2 relative to node k}
	\begin{split}
	\overline{J}''_{\C(k)}&(\bw_k)
	=\sum_{\ell\in\N_k\cap\C(k)}c_{\ell k}\expec\big\{|d_\ell(i)-\bx_{\ell}^\top(i)\bw_k |^2\big\}\\
	&+2\eta\hspace{-2mm}\sum_{\ell\in\N_k\setminus\C(k)}\hspace{-0.3cm}p_{k\ell}f(\bw_k-\bw_\ell) 
	+\hspace{-3mm}\sum_{\ell\in\N_k^-\cap\C(k)}\hspace{-0.3cm}b_{\ell k}\|\bw_k-\bw_\ell^{\text{loc}}\|^2.
	\end{split}
\end{equation}
Since this cost function only relies on data available in the neighborhood of each node $k$, we can now proceed to derive distributed strategies.

The first and third terms on the RHS of \eqref{eq: cluster function 2 relative to node k} are second-order differentiable and strictly convex. The second term is convex but not continuously differentiable. In~\cite{nassif2015icassp}, a multitask Adapt-then-Combine (ATC) diffusion algorithm was derived using subgradient techniques. The purpose of this work is to obtain an iterative algorithm for solving the convex minimization problem~\eqref{eq: cluster function 2 relative to node k} using a forward-backward splitting approach.

\subsection{Multitask diffusion with forward-backward splitting approach}

Let $\bw_k(i)$ denote the estimate of $\bw^o_k$ at node $k$ and iteration $i$. Considering a forward-backward splitting strategy for solving~\eqref{eq: cluster function 2 relative to node k}, we have:
\begin{equation}
\label{eq: general splitting approach}
	\bw_k(i+1)=\text{prox}_{2\eta\nu_k \tilde{g}_{k,i}}\Big(\bw_k(i)-\nu_k\nabla_{\bw_k} J''_{\C(k)}(\bw_k(i))\Big),
\end{equation}
with $\nu_k$ a positive step-size parameter,
\begin{equation}
	\label{eq: g_k form 1}
	\tilde{g}_{k,i}(\bw_k)=\sum_{\ell\in\N_k\setminus\C(k)}p_{k\ell}f(\bw_k-\bw_\ell(i)),
\end{equation}
and $J''_{\C(k)}(\bw_k)$ denoting the unregularized part of $\overline{J}''_{\C(k)}(\bw_k)$ limited to the first and third terms on the RHS of~\eqref{eq: cluster function 2 relative to node k}. Let
\begin{equation}
\bphi_k(i+1)=\bw_k(i)-\nu_k\nabla_{\bw_k}  J''_{\C(k)}(\bw_k(i)).
\end{equation}
 Node $k$ can run the Adapt-then-Combine (ATC) form of diffusion~\cite{sayed2014diffusion} for evaluating $\bphi_k(i+1)$. Thus, we arrive at the following Adapt-then-Combine (ATC) diffusion strategy with forward-backward splitting for solving problem~\eqref{eq: problem P} in a fully distributed adaptive manner:
\begin{equation}
	\label{eq: ATC FBS}
	\left\lbrace
	\begin{array}{lr}
	\begin{split}
	\bpsi_k(i+1)= &\bw_k(i)+\\
			  &\mu_k\hspace{-3mm}\sum\limits_{\ell\in\N_k\cap\C(k)}\hspace{-2mm}c_{\ell k}\,\bx_{\ell}(i)[d_{\ell}(i)-\bx_{\ell}^\top(i)\bw_k(i)],
	\end{split}\\
	\bphi_k(i+1)=\sum\limits_{\ell\in\N_k\cap\C(k)}a_{\ell k}\bpsi_\ell(i+1),\\
	\bw_k(i+1)=\text{prox}_{\eta\mu_k g_{k,i+1}}(\bphi_k(i+1)),
	\end{array}
	\right.
\end{equation}
where $\mu_k=2\nu_k$ is introduced to avoid an extra factor of $2$ multiplying $\nu_k$ and coming from evaluating the gradient of squared quantities in $J''_{\C(k)}(\bw_k)$, $\{a_{\ell k}\}$ are nonnegative combination coefficients satisfying:
\begin{equation}
	\label{eq: conditions on a}
	\sum_{\ell=1}^N a_{\ell k}=1,\quad\text{and}\quad a_{\ell k}=0\quad \text{if}\quad \ell\notin\N_k\cap\C(k),
\end{equation}
and 
\begin{equation}
	\label{eq: g_k form 2}
	g_{k,i+1}(\bw_k)\triangleq\sum_{\ell\in\N_k\setminus\C(k)}p_{k\ell}f(\bw_k-\bphi_\ell(i+1)).
\end{equation}
Functions $\tilde{g}_{k,i}(\cdot)$ in~\eqref{eq: g_k form 1} and $g_{k,i+1}(\cdot)$ in~\eqref{eq: g_k form 2} are iteration dependent through $\bw_\ell(i)$ and $\bphi_\ell(i+1)$. Note that we have substituted $\bw_{\ell}(i)$ in~\eqref{eq: g_k form 1} by $\bphi_{\ell}(i+1)$ in~\eqref{eq: g_k form 2} since $\bphi_{\ell}(i+1)$ is an updated estimate of $\bw_{\ell}(i)$ at node~$\ell$. The proximal operator of $\eta\mu_k g_{k,i+1}(\cdot)$ in the third step of~\eqref{eq: ATC FBS} needs to be evaluated at each iteration $i+1$ and for all nodes $k$ in the network. A closed-form expression is recommended to achieve higher computational efficiency. We shall derive such closed-form expression when $f$ in~\eqref{eq: g_k form 2} is selected either as the $\ell_1$-norm or the reweighted $\ell_1$-norm --- see Sec.~\ref{subsec: proximal operator} for details.

The multitask diffusion LMS~\eqref{eq: ATC FBS} with forward-backward splitting starts with an initial estimate $\bw_k(0)$ for all $k$, and repeats~\eqref{eq: ATC FBS} at each instant $i\geq 0$ and for all $k$. In the first step of~\eqref{eq: ATC FBS}, which corresponds to the adaptation step, node $k$ receives from its intra-cluster neighbors their raw data $\{d_{\ell}(i),\bx_{\ell}(i)\}$, combines this information through the coefficients $\{c_{\ell k}\}$, and uses it to update its estimate $\bw_k(i)$ to an intermediate estimate $\bpsi_{k}(i+1)$. The second step in~\eqref{eq: ATC FBS} is a combination step where node $k$ receives the intermediate estimates $\{\bpsi_{\ell}(i+1)\}$ from its intra-cluster neighbors and combines them through the coefficients $\{a_{\ell k}\}$ to obtain the intermediate value $\bphi_{k}(i+1)$. Finally, in the third step in~\eqref{eq: ATC FBS}, node $k$ receives the intermediate estimates $\{\bphi_{\ell}(i+1)\}$ from its neighbors that are outside its cluster and evaluates the proximal operator of the function in~\eqref{eq: g_k form 2} at $\bphi_k(i+1)$ to obtain $\bw_k(i+1)$. To run the algorithm, each node $k$ only needs to know the step-size $\mu_k$, the regularization strength $\eta$, the regularization weights $\{p_{k\ell}\}_{\ell\in\N_k\setminus\C(k)}$, and the coefficients $\{a_{\ell k},c_{\ell k}\}_{\ell\in\N_k\cap\C(k)}$ satisfying conditions~\eqref{eq: conditions on c} and~\eqref{eq: conditions on a}. The scalars $\{a_{\ell k},c_{\ell k}\}$ and $\{\rho_{k \ell}\}$ correspond to weighting coefficients over the edges linking node $k$ to its neighbors $\ell$ according to whether these neighbors lie inside or outside its cluster. There are several ways to select these coefficients~\cite{sayed2014adaptivenet,sayed2014diffusion,chen2014multitask,sayed2014adaptation}. In Section~\ref{sec: simulation results}, we propose an adaptive rule for selecting each regularization weight $p_{k\ell}$ based on a measure of the sparsity level of $\bw_k^o-\bw_\ell^o$ at node $k$. Finally, note that alternative implementations of~\eqref{eq: ATC FBS} may be considered. In particular, the adaptation step can be followed by the proximal step, before or after aggregation as in the possible Adapt-then-Combine and Combine-then-Adapt diffusion strategies. 

Algorithm~\eqref{eq: ATC FBS} may be applied to multitask problems involving any type of coregularizers $f(\cdot)$ provided that the proximal operator of a weighted sum of these regularizers can be assessed in closed form. In the next section, we shall focus on the particular case of sparsity promoting regularizers.

\subsection{Proximal operator of weighted sum of $\ell_1$-norms}
\label{subsec: proximal operator}

We shall now derive a closed form expression for the proximal operator of the convex function $g_{k,i+1}(\bw_k)$ in~\eqref{eq: g_k form 2}. Considering both regularizations addressed in this work, that is, the $\ell_1$-norm~\eqref{eq: l1 norm} and the reweighted $\ell_1$-norm~\eqref{eq: reweighted l1 norm},  we write:
\begin{eqnarray}
	\label{eq: g_k form 3}
	g_{k,i+1}(\bw_k)\hspace{-3mm}&=&\hspace{-4mm}\sum\limits_{\ell\in\N_k\setminus\C(k)}p_{k\ell}\sum\limits_{m=1}^M
	\alpha_{k\ell}^m(i)\big|[\bw_k]_m-[\bphi_\ell(i+1)]_m\big|\nonumber\\
	 &=&\sum\limits_{m=1}^M\Phi_{{k},m,i+1}([\bw_k]_m)
\end{eqnarray}
where $\Phi_{{k},m,i+1}([\bw_k]_m)$ is the iteration-dependent function given by:
\begin{equation}
	\label{eq: Phi_km}
	\Phi_{{k},m,i+1}([\bw_k]_m)=\hspace{-2mm}\sum\limits_{\ell\in\N_k\setminus\C(k)}p_{k\ell}\,
	\alpha_{k\ell}^m(i)\big|[\bw_k]_m-[\bphi_\ell(i+1)]_m\big|.
\end{equation}
Since $g_{k,i+1}(\bw_k)$ is fully separable, its proximal operator can be evaluated component-wise~\cite{parikh2013proximal}:
\begin{equation}
	\label{eq: fully separable prox}
	\begin{split}
	[&\text{prox}_{\eta\mu_k g_{k,i+1}}(\bphi_k(i+1))]_m\\
	&=\text{prox}_{\eta\mu_k\Phi_{{k},m,i+1}}([\bphi_k(i+1)]_m),\quad\forall m=1,\ldots,M.
	\end{split}
\end{equation}
For clarity of presentation, we shall now derive the proximal operator of a function $h(\cdot)$ similar to $\Phi_{{k},m,i+1}$.
Next, we shall establish the closed-form expression for $\text{prox}_{\eta\mu_k\Phi_{{k},m,i+1}}(\cdot)$ by identification. 

Let $h:\mathbb{R}\rightarrow\mathbb{R}$ be a combination of absolute value functions defined as:
\begin{equation}
	\label{eq: func sum absolute}
	h(x)\triangleq\sum\limits_{j=1}^J c_j\,h_j(x)=\sum\limits_{j=1}^J c_j\, |x-b_j|,
\end{equation}
with $c_j>0$ for all $j$ and $b_1<b_2<\ldots<b_J$. Note that this ordering is assumed for convenience of derivation and does not affect the final result. Iterative algorithms have been proposed in the literature for evaluating the proximal operator of sums of composite functions~\cite{combettes2011proximal,combettes2011proximity}. We are, however, able to derive a closed-form expression for~\eqref{eq: func sum absolute} as detailed in the sequel. From the optimality condition for \eqref{eq: proximal operator}, namely that zero belongs to the subgradient set at the minimizer $\prox_{\lambda h}(v)$, we have,
\begin{equation}
	\label{eq: optimality condition}
	\begin{split}
	&0\in\partial h(\prox_{\lambda h}(v))+\frac{1}{\lambda}(\prox_{\lambda h}(v)-v)\\
	&\Rightarrow	v-\prox_{\lambda h}(v) \in\lambda\partial h(\prox_{\lambda h}(v)).
	\end{split}
\end{equation}
Since $x\in \mathbb{R}$ and $c_j$ are non-negative, we have~\cite[Chapter 5: Lemma 10]{polyak1987optim}:
\begin{equation}
	\label{eq: sum of partial derivatives}
	\partial\Big(\sum\limits_{j=1}^J c_j h_j(x)\Big)=\sum\limits_{j=1}^J c_j \partial h_j(x)=\sum\limits_{j=1}^J c_j \partial |x-b_j|.
\end{equation}
Hence, the subdifferential of the real valued convex function $h(x)$ in~\eqref{eq: func sum absolute} is:
\begin{equation}
\label{eq: subdifferential h}
\partial h(x)=\left\lbrace
\begin{split}
&-\sum_{j=1}^J c_j,~\quad\quad\quad\quad\quad\text{if}\quad x<b_1,\\
&c_1\cdot[-1,1]-\sum_{j=2}^J c_j,\quad\text{if}\quad x=b_1,\\
&c_1-\sum_{j=2}^J c_j,~~\quad\quad\quad\quad\text{if}\quad b_1<x<b_2,\\
&\quad\quad\vdots\\
&\sum_{j=1}^{J-1} c_j +c_J\cdot[-1,1],\quad\text{if}\quad x=b_J,\\
&\sum_{j=1}^{J} c_j,\quad\quad\quad\quad\quad\quad \quad\text{if}\quad x>b_J.
\end{split}
\right.
\end{equation}
From~\eqref{eq: optimality condition} and~\eqref{eq: subdifferential h}, extensive but routine calculations lead to the following implementation for evaluating the proximal operator of $h$ in~\eqref{eq: func sum absolute}. Let us decompose $\mathbb{R}$ into $J+1$ intervals such that $\mathbb{R}=\bigcup\limits_{n=0}^J \cI_{n}$ where, as illustrated in Fig.~\ref{fig: intervals}:
\begin{eqnarray}
\label{eq: I0}\cI_{0}&\triangleq&\Big]-\infty \, ,\, b_1-\lambda\sum\limits_{j=1}^J c_j\Big[,\label{eq: I_0}\\
\label{eq: In}\cI_{n}&\triangleq&\cI_{n,1}\cup\cI_{n,2},\qquad n=1,\ldots, J,\label{eq: I_n}
\end{eqnarray}
with
\begin{eqnarray}
\label{eq: In1}\cI_{n,1}\hspace{-3mm}&\triangleq&\hspace{-3.25mm}\Big[ b_n-\lambda\Big(\sum\limits_{j=n}^J c_j-\hspace{-0.5mm}\sum\limits_{j=1}^{n-1}c_j\Big)  ,  b_n\hspace{-0.5mm}-\lambda\Big(\hspace{-1mm}\sum\limits_{j=n+1}^J\hspace{-1mm} c_j-\hspace{-0.5mm}\sum\limits_{j=1}^{n}c_j\Big)\Big[,\nonumber\\
		\hspace{-3mm}& &\hspace{-3.25mm}\qquad \qquad\qquad\quad n=1,\ldots,J,\\
\label{eq: In2}\cI_{n,2}\hspace{-3mm}&\triangleq&\hspace{-3.25mm}\Big[b_n\hspace{-0.5mm}-\hspace{-0.25mm}\lambda\Big(\hspace{-1mm}\sum\limits_{j=n+1}^J \hspace{-2.25mm}c_j\hspace{-0.25mm}-\hspace{-0.5mm}\sum\limits_{j=1}^{n}c_j\Big) , b_{n+1}\hspace{-0.5mm}-\hspace{-0.25mm}\lambda\Big(\hspace{-1.25mm}\sum\limits_{j=n+1}^J\hspace{-2.25mm} c_j\hspace{-0.5mm}-\hspace{-0.5mm}\sum\limits_{j=1}^{n}c_j\Big)\Big[,\nonumber\\		
		\hspace{-3mm}& &\hspace{-3.25mm}\qquad \qquad\qquad\quad n=1,\ldots,J-1,\\
\label{eq: IJ2}\cI_{J,2}\hspace{-3mm}&\triangleq&\hspace{-3.25mm}\Big[ b_J+\lambda\sum\limits_{j=1}^J c_j  , +\infty\Big[.
\end{eqnarray}
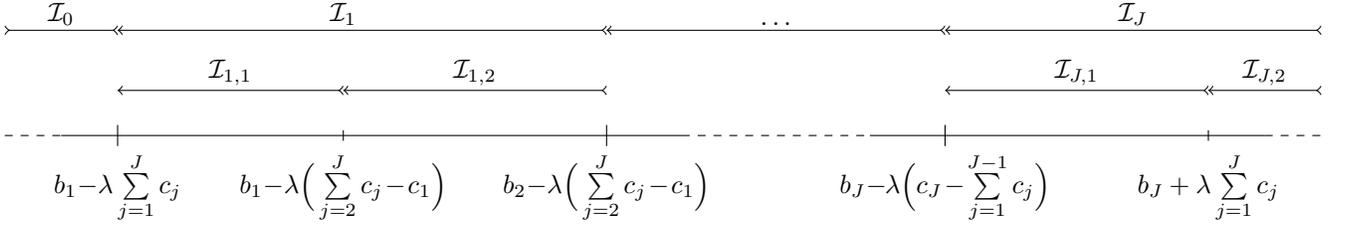
\begin{figure*}
\centering
\begin{tikzpicture}
    \draw [dashed](0,0) -- (0.75,0);
    \draw (0.75,0) -- (1.5,0);
    \draw (1.5,0) -- (4.5,0);
     \draw (4.5,0) -- (8,0);
      \draw (8,0) -- (9,0);
    \draw [dashed] (9,0) -- (11.5,0);
    \draw (11.5,0) -- (12.5,0);
     \draw (12.5,0) -- (16,0);
      \draw (16,0) -- (16.75,0);
       \draw [dashed](16.75,0) -- (17.5,0);
    
  \draw [>-<](0,1.4) -- (1.5,1.4);
   \draw  [<-<](1.5,1.4) -- (8,1.4);
   \draw  [<-<](1.5,0.6) -- (4.5,0.6);
     \draw [<-<](4.5,0.6) -- (8,0.6);
    \draw [<-<](12.5,0.6) -- (16,0.6);
     \draw [<-<](16,0.6) -- (17.5,0.6);
     \draw [<-<](12.5,1.4) -- (17.5,1.4);
      \draw [<-<](8,1.4) -- (12.5,1.4);
      
    \foreach \x in {1.5,8,12.5}
      \draw (\x cm,4pt) -- (\x cm,-4pt);
     \foreach \x in {4.5,16}
     \draw (\x cm,2pt) -- (\x cm,-2pt);

   \draw (1.5,0) node[below=3pt] {$b_1\hspace{-0.75mm}-\hspace{-0.75mm}\lambda\sum\limits_{j=1}^Jc_j $};
    \draw (4.5,0) node[below=3pt] {$ b_1\hspace{-0.75mm}-\hspace{-0.75mm}\lambda\Big(\sum\limits_{j=2}^Jc_j\hspace{-0.5mm}-\hspace{-0.5mm}c_1\Big) $};
    \draw (8,0) node[below=3pt] {$ b_2\hspace{-0.75mm}-\hspace{-0.75mm}\lambda\Big(\sum\limits_{j=2}^Jc_j \hspace{-0.5mm}-\hspace{-0.5mm}c_1\Big)$};
     \draw (12.5,0) node[below=3pt] {$ b_J\hspace{-0.75mm}-\hspace{-1mm}\lambda\Big(c_J\hspace{-0.5mm}-\hspace{-0.5mm}\sum\limits_{j=1}^{J-1}c_j\Big)$};
     \draw (16,0) node[below=3pt] {$b_J+\lambda\sum\limits_{j=1}^{J}c_j$} node[above=3pt] {$  $};
    \draw (0.75,1.35) node[above=0pt] {$\cI_0$} ;
    \draw (4.5,1.35) node[above=0pt] {$\cI_1$} ;
     \draw (3,0.55) node[above=0pt] {$\cI_{1,1}$} ;
     \draw (6.25,0.55) node[above=0pt] {$\cI_{1,2}$} ;
      \draw (10.25,1.35) node[above=0pt] {$\ldots$} ;
     \draw (14.25,0.55) node[above=0pt] {$\cI_{J,1}$} ;
      \draw (16.75,0.55) node[above=0pt] {$\cI_{J,2}$} ;
    \draw (15,1.35) node[above=0pt] {$\cI_J$} ;
 \end{tikzpicture}
\caption{Decomposition of $\mathbb{R}$ into $J+1$ intervals given by~\eqref{eq: I_0}--\eqref{eq: IJ2}. The width of the intervals depends on the weights $\{c_j\}_{j=1}^J$ and on the coefficients $\{b_j\}_{j=1}^J$.}
\label{fig: intervals}
\end{figure*}

Depending on the interval to which $v$ belongs, we evaluate the proximal operator according to:
\begin{equation}
	\label{eq: prox sum absolute}
	\text{prox}_{\lambda h}(v)
	=\left\lbrace
	\begin{array}{ll}
		v+\lambda\sum\limits_{j=1}^J c_j,	&	\text{if}~ v\in\cI_0	\\		

		b_n,	 &	\text{if } v\in\cI_{n,1}\\
		
		v+\lambda\Big(\sum\limits_{j=n+1}^J c_j-\sum\limits_{j=1}^{n}c_j\Big),& \text{if}~ v\in\cI_{n,2}.
	\end{array}
\right.
\end{equation}
In order to make clearer how the operator in~\eqref{eq: prox sum absolute} works, we plot $\prox_{h}(v)$ for three expressions of $h$ in Fig.~\ref{fig: prox_sum_abs}.
\begin{figure*}
	\centering
	\subfigure[$h(x)=|x|$ ]{\includegraphics[width=0.3\textwidth, height=0.18\textheight]{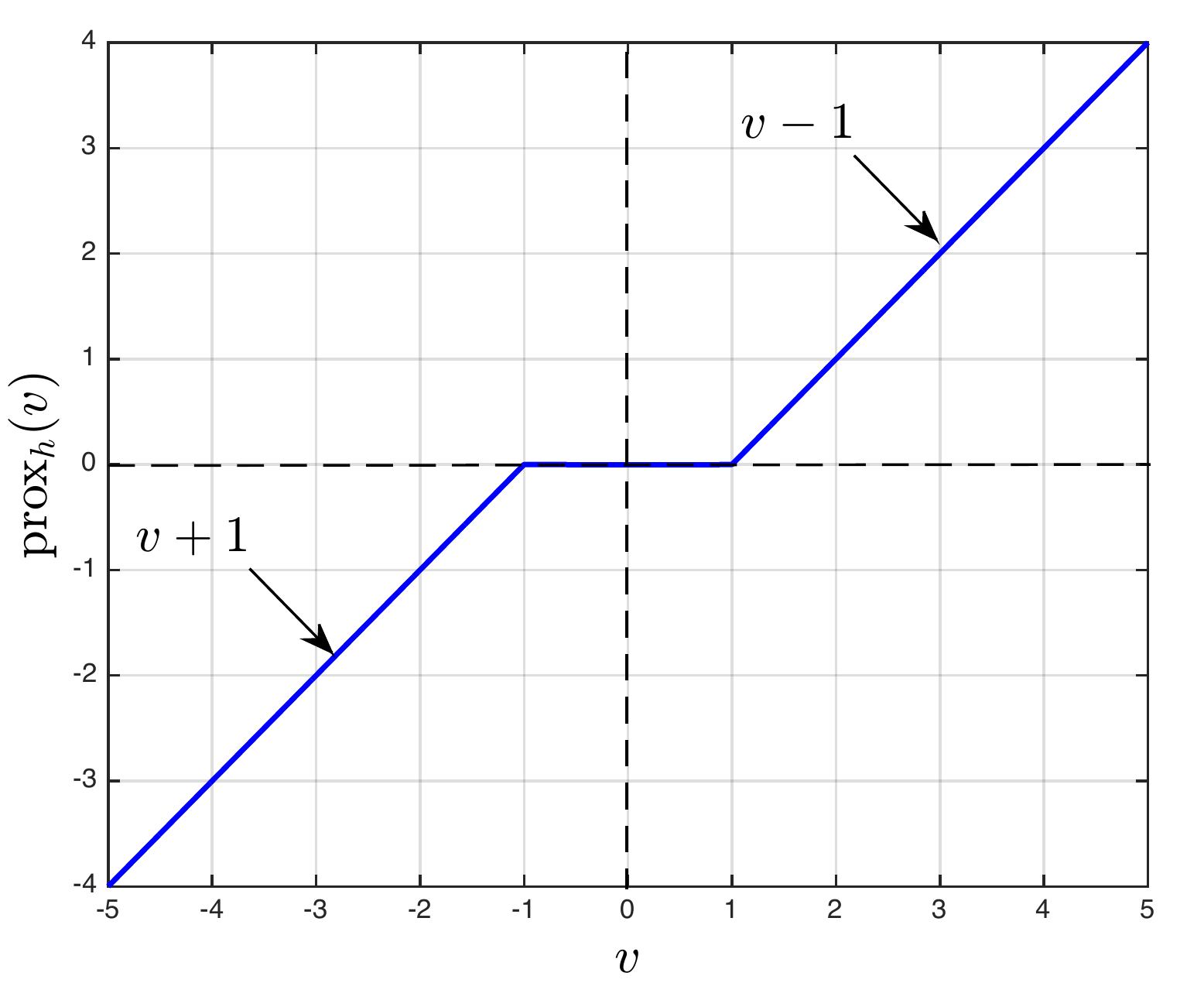}}
	\subfigure[$h(x)=\frac{1}{2}|x-2|$]{\includegraphics[width=0.3\textwidth,height=0.18\textheight]
	{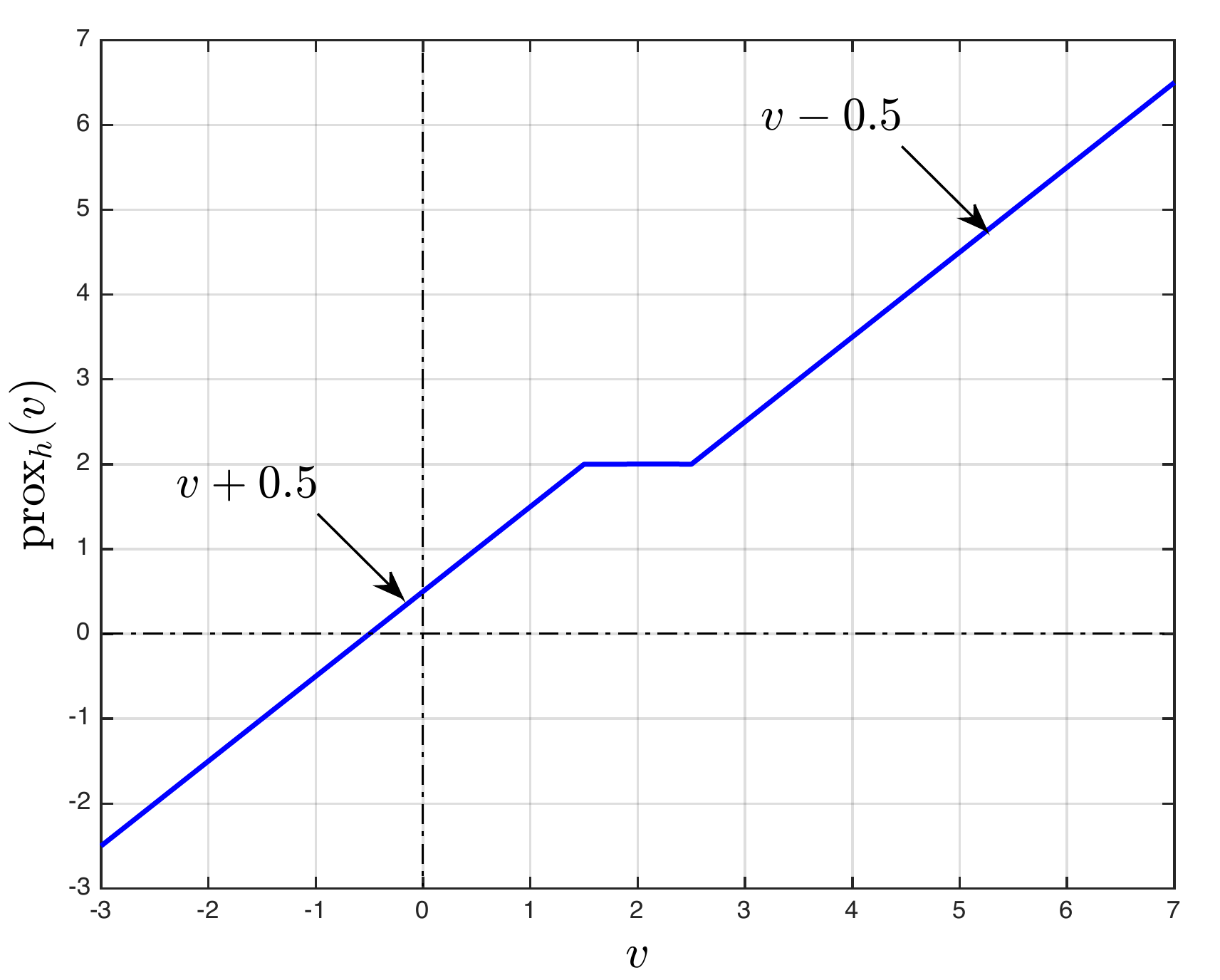}}
	\subfigure[$h(x)=\frac{1}{5}|x+1|+\frac{3}{10}|x-1|+\frac{1}{2}|x-4|$]
	{\includegraphics[width=0.3\textwidth,height=0.18\textheight]{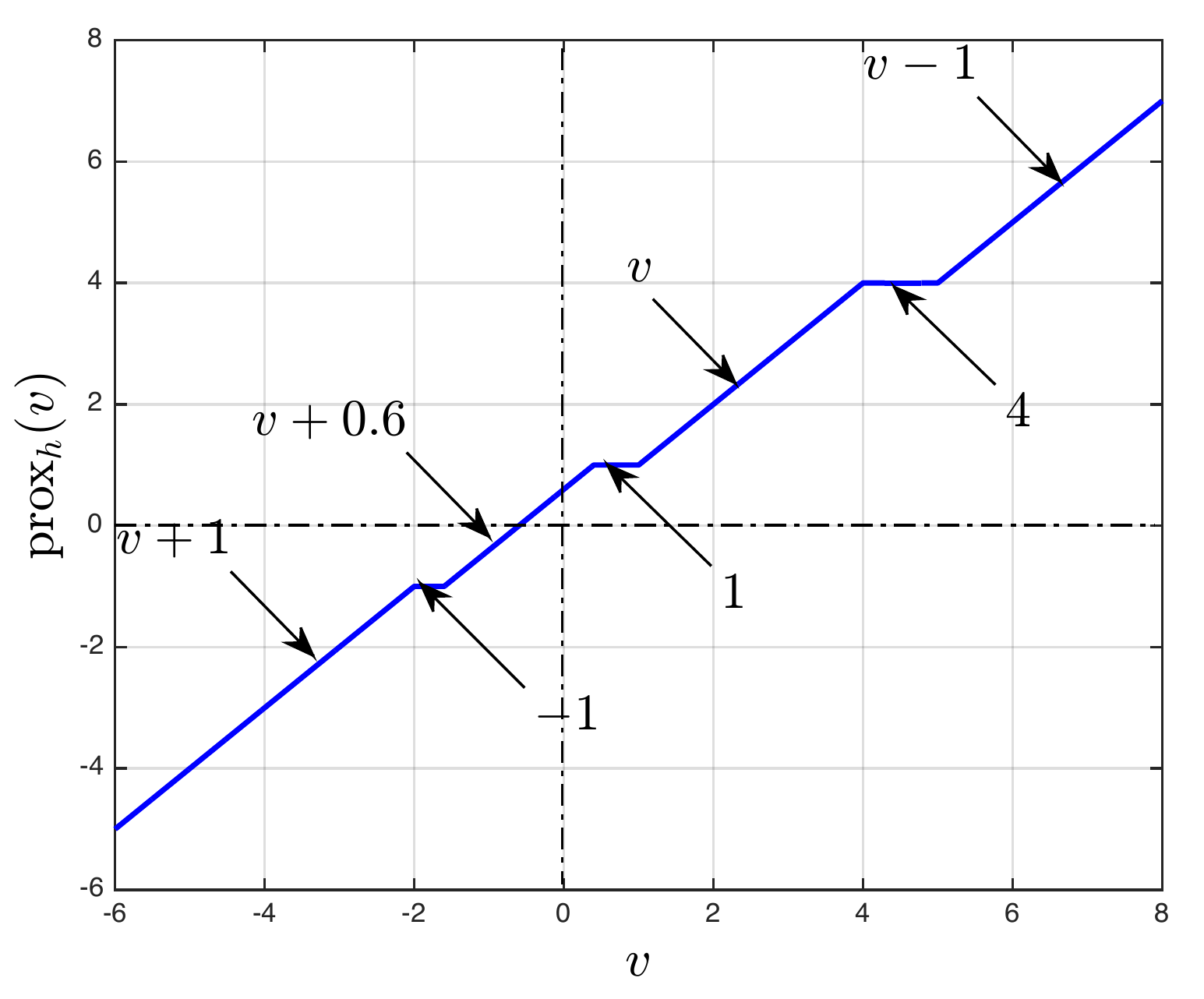}}
	\caption{Proximal operator $\prox_{\lambda h}(v)$ versus $v\in\mathbb{R}$ with $\lambda=1$ 
	and $h: \mathbb{R}	\rightarrow\mathbb{R}$, $h(x)=\sum\limits_{j=1}^J c_j |x-b_j|$.}
	\label{fig: prox_sum_abs}
\end{figure*}

It can be checked that the proximal operator in~\eqref{eq: prox sum absolute} can be written more compactly as:
\begin{equation}
	\label{eq: prox casting}
	\text{prox}_{\lambda h}(v)=v-\lambda\Gamma(v),
\end{equation} 
where
\begin{equation}
\begin{split}
	\label{eq: limiter function}
	&\Gamma(v)=\\
	&\frac{1}{2}\hspace{-0.5mm}\sum\limits_{n=1}^J\hspace{-0.5mm}\Big\{\Big|\frac{v-b_n}{\lambda}\hspace{-0.5mm}-\hspace{-0.5mm}\sum_{j=1}^{n-1}\hspace{-0.25mm}c_j\hspace{-0.5mm}+\hspace{-0.75mm}\sum_{j=n}^J\hspace{-0.25mm}c_j\Big|\hspace{-0.5mm}-\hspace{-0.5mm}\Big|\frac{v-b_n}{\lambda}\hspace{-0.5mm}-\hspace{-0.5mm}\sum_{j=1}^{n}c_j\hspace{-0.5mm}+\hspace{-2.5mm}\sum_{j=n+1}^J\hspace{-2.25mm}c_j\Big|\Big\}.
\end{split}
\end{equation}
Comparing~\eqref{eq: optimality condition} and~\eqref{eq: prox casting}, we remark that $\Gamma(v)$ is a subgradient of $h$ at $\prox_{\lambda h}(v)$. 
Based on equation~\eqref{eq: prox sum absolute}, $\Gamma(v)$ is bounded as follows:
\begin{equation}
	\label{eq: boundedness property}
	|\Gamma(v)|\leq\sum_{j=1}^J c_j
\end{equation}
for all $v$. In fact, equality holds when $v$ belongs to $\cI_0$ in~\eqref{eq: I0} or $\cI_{J,2}$ in~\eqref{eq: IJ2}. When $v$ belongs to an interval of the form of $\cI_{n,1}$ in~\eqref{eq: In1}, we have:
\begin{equation}
\begin{split}
\Gamma(v)=\frac{v-b_n}{\lambda}\in\Big[&\sum\limits_{j=1}^{n-1}c_j-\sum\limits_{j=n}^J c_j , \sum\limits_{j=1}^{n}c_j-\sum\limits_{j=n+1}^J c_j\Big]\\
&\subset\Big[-\sum\limits_{j=1}^J c_j,\sum\limits_{j=1}^J c_j\Big],
\end{split}
\end{equation} 
and when it belongs to an interval of the form of $\cI_{n,2}$ in~\eqref{eq: In2}, we have:
\begin{equation}
\Gamma(v)=\sum\limits_{j=1}^{n}c_j-\sum\limits_{j=n+1}^J c_j\in\Big[-\sum\limits_{j=1}^J c_j,\sum\limits_{j=1}^J c_j\Big].
\end{equation} 
We note that the upper bound in $\eqref{eq: boundedness property}$ is independent of $\lambda$. 
Using~\eqref{eq: prox casting}, the $m$-th entry of $\prox_{\eta\mu_k g_{k,i+1}}(\bphi_k(i+1))$ in~\eqref{eq: fully separable prox} can be written as:
\begin{equation}
\begin{split}
	\label{eq: prox casting 2}
	[&\prox_{\eta\mu_k g_{k,i+1}}(\bphi_k(i+1))]_m\\
	&=[\bphi_k(i+1)]_m-\eta\mu_k\Gamma_{k,m,i+1}([\bphi_k(i+1)]_m).
\end{split}
\end{equation}
Note that $\Gamma_{k,m,i+1}([\bphi_k(i+1)]_m)]$ is a function of the form~\eqref{eq: limiter function} where, based on~\eqref{eq: Phi_km}, the coefficients $b_j$ and $c_j$ are given by $[\bphi_\ell(i+1)]_m$ and $p_{k\ell}\,\alpha_{k\ell}^m(i)$, respectively, and the scalar $v$ corresponds to the $m$-th component of the vector $\bphi_k(i+1)$. 
Using the boundedness of $\Gamma_{k,m,i+1}(\cdot)$ in~\eqref{eq: boundedness property}, we obtain:
\begin{equation}
	\label{eq: bound on Gamma}
	|\Gamma_{k,m,i+1}([\bphi_k(i+1)]_m)|\leq\sum_{\ell\in\N_k\setminus\C(k)}p_{k\ell}\,\alpha_{k\ell}^m(i)\,\triangleq s_k^m(i)
\end{equation}
for all $[\bphi_k(i+1)]_m$. For the $\ell_1$-norm~\eqref{eq: l1 norm}, we have:
\begin{equation}
	s_k^m(i)=s_k\triangleq\sum_{\ell\in\N_k\setminus\C(k)}p_{k\ell},
\end{equation}
for all $i$ and $m=1,\ldots, M$. For the reweighted $\ell_1$-norm~\eqref{eq: reweighted l1 norm}, we have:
\begin{align}
	s_k^m(i)&=\sum\limits_{\ell\in\N_k\setminus\C(k)}\frac{p_{k\ell}}
	{\epsilon+|[\bdelta_{k,\ell}(i-1)]_m|}\nonumber\\
	&=\frac{1}{\epsilon}\sum\limits_{\ell\in\N_k\setminus\C(k)}\frac{p_{k\ell}}
	{1+\frac{|[\bdelta_{k,\ell}(i-1)]_m|}{\epsilon}}\nonumber\\
	&\leq \frac{s_k}{\epsilon}
\end{align}
for all $i$ and $m=1,\ldots, M$. Using~\eqref{eq: prox casting 2}, the proximal operator of $\eta\mu_k g_{k,i+1}$ can be written as:
\begin{equation}
	\label{eq: prox casting 3}
	\prox_{\eta\mu_k g_{k,i+1}}(\bphi_k(i+1))=\bphi_k(i+1)-\eta\mu_k\bGamma_{k,i+1}(\bphi_k(i+1)),
\end{equation}
where $\bGamma_{k,i+1}(\bphi_k(i+1))$ is the $M\times1$ vector given by:
\begin{equation}
\begin{split}
	\label{eq: vector Gamma_{k,i+1}}
	&\bGamma_{k,i+1}(\bphi_k(i+1))\\
	&=\col\Big\{\Gamma_{k,1,i+1}([\bphi_k(i+1)]_1,\ldots,
	\Gamma_{k,M,i+1}([\bphi_k(i+1)]_M)\Big\}.
\end{split}
\end{equation}
As a consequence, the $\ell_2$-norm of the vector $\bGamma_{k,i+1}(\cdot)$ can be bounded as:
\begin{eqnarray}
	\|\bGamma_{k,i+1}(\cdot)\|_2 \hspace{-2mm}&\leq& \hspace{-2mm}s_k\sqrt{M},\text{for the }\ell_1\text{-norm},\label{eq: bound l1}\\
	\|\bGamma_{k,i+1}(\cdot)\|_2 \hspace{-2mm}&\leq& \hspace{-2mm}\frac{s_k\sqrt{M}}{\epsilon},\text{for the reweighted }\ell_1\text{-norm}.\label{eq: bound rew l1}
\end{eqnarray}


\section{Stability analysis}
\label{sec: Mean-Square performance analysis}

\subsection{Error vector recursion}

We shall now analyze the stability of the multitask diffusion algorithm~\eqref{eq: ATC FBS} in the mean and mean-square-error sense. We first define at node $k$ and iteration $i$ the weight error vector $\bwt_k(i)\triangleq\bw^o_k-\bw_k(i)$ and the intermediate error vector $\bphit_k(i)\triangleq\bw^o_k-\bphi_k(i)$. Furthermore, we introduce the network vectors:
\begin{eqnarray}
	\bwt(i)	&\triangleq	&\col\left\{\bwt_1(i),\ldots,\bwt_N(i)\right\}\label{eq: network error vector 1}\\
	\bphi(i)	&\triangleq	&\col\left\{\bphi_1(i),\ldots,\bphi_N(i)\right\}\\
	\bphit(i)	&\triangleq	&\col\left\{\bphit_1(i),\ldots,\bphit_N(i)\right\}.
\end{eqnarray}
Let $\cM$ and $\cR_{\bx}(i)$ be the $MN \times MN$ block diagonal matrices defined as:
\begin{eqnarray}
	\cM&\triangleq&\diag\left\{\mu_k\bI_M\right\}_{k=1}^N \\
	\cR_{\bx}(i)&\triangleq&\diag\,\Big\{\sum_{\ell\in\N_k\cap\C(k)}c_{\ell k}\,\bx_{\ell}(i)\,\bx^\top_{\ell}(i)\Big\}_{k=1}^N
\end{eqnarray}
and $\bp_{zx}(i)$ be the $MN \times 1$ block vector defined as:
\begin{equation}
	\bp_{zx}(i)
	\triangleq\cCT\col\big\{\bx_{k}(i)\,z_{k}(i)\big\}_{k=1}^N,
\end{equation}
where $\cC\triangleq\bC\otimes\bI_M$ and $\bC$ is the $N\times N$ right-stochastic matrix whose $\ell k$-th entry is $c_{\ell k}$. Let $\cA\triangleq\bA\otimes\bI_M$ where $\bA$ is the $N\times N$ left-stochastic matrix  whose $\ell k$-th entry is $a_{\ell k}$. Subtracting $\bw^o_k$ from both sides of the first and second step in~\eqref{eq: ATC FBS}, and using the linear data model \eqref{eq: linear data model}, we obtain:
\begin{equation}
	\label{eq: error recursion phi}
	\bphit(i+1)=\cAT[\bI_{MN}-\cM\cR_{\bx}(i)]\bwt(i)-\cAT\cM\bp_{z\bx}(i).
\end{equation} 
Subtracting $\bw^o_k$ from both sides of the third step in~\eqref{eq: ATC FBS}, and using result~\eqref{eq: prox casting 3}, we get:
\begin{equation}
	\bwt_k(i+1)=\bphit_k(i+1)+\eta\mu_k\,\bGamma_{k,i+1}(\bphi_k(i+1)).
\end{equation}
Hence, the network error vector for the diffusion strategy~\eqref{eq: ATC FBS} evolves according to the following recursion:
\begin{equation}
	\label{eq: network error vector recursion}
	\fbox{$\begin{split}
	\bwt(i+1)=\cAT[\bI_{MN}-&\cM\cR_{\bx}(i)]\,\bwt(i)-\cAT\cM\bp_{z\bx}(i)+\\
	&\eta\,\cM\bGamma_{i+1}(\bphi(i+1)),
	\end{split}$}
\end{equation}
where $\bGamma_{i+1}(\bphi(i+1))$ is the $N\times 1$ block vector whose $k$-th block is given by~\eqref{eq: vector Gamma_{k,i+1}}, namely,
\begin{equation}
	\label{eq: Gamma i}
	\bGamma_{i+1}(\bphi(i+1))\triangleq\col\Big\{\bGamma_{k,i+1}(\bphi_k(i+1))\Big\}_{k=1}^N.
\end{equation}
In order to make the presentation clearer, we shall use the following notation for terms in recursion~\eqref{eq: network error vector recursion}:
\begin{eqnarray}
	\cB(i)&\triangleq	&\cAT[\bI_{MN}-\cM\cR_{\bx}(i)],\\
	\bg(i)&\triangleq	&\cAT\cM\bp_{z\bx}(i), \\
	\br(i+1)&\triangleq	&\eta\,\cM\bGamma_{i+1}(\bphi(i+1)).
\end{eqnarray}
Hence, recursion~\eqref{eq: network error vector recursion} can be rewritten as follows:
\begin{equation}
	\label{eq: network error vector recursion 2}
	\bwt(i+1)=\cB(i)\bwt(i)-\bg(i)+\br(i+1).
\end{equation}

Before proceeding, let us introduce the following assumptions on the regression data and  step-sizes.
\begin{assumption}{(Independent regressors)}
\label{assumption: independent regressors}
The regression vectors $\bx_k(i)$ arise from a zero-mean random process that is temporally white and spatially independent.
\end{assumption}
It follows that $\bx_k(i)$ is independent of $\bw_\ell(j)$ for $i\geq j$ and for all $\ell$. This assumption is commonly used in adaptive filtering since it helps simplify the analysis. Furthermore, performance results obtained under this assumption match well the actual performance of stand alone filters for sufficiently small step-sizes~\cite{sayed2011adaptive}.
\begin{assumption}{(Small step-sizes)}
\label{assumption: step-sizes assumption}
The step-sizes $\mu_k$ are sufficiently small so that terms that depend on higher order powers of the step-sizes can be ignored.
\end{assumption}


\subsection{Mean behavior analysis}

Taking the expectation of both sides of~\eqref{eq: network error vector recursion 2}, using Assumption~\ref{assumption: independent regressors}, and $\expec\{\bp_{z\bx}(i)\}=0$, we obtain that the mean error vector evolves according to the following recursion:
\begin{equation}
	\label{eq: mean error recursion}
	\expec\{\bwt(i+1)\}=\cB\,\expec\{\bwt(i)\}+\expec\{\br(i+1)\},
\end{equation}
where
\begin{eqnarray}
\cB\hspace{-2mm}&\triangleq\hspace{-2mm}&\cAT(\bI_{MN}-\cM\cR_{\bx}),\\
\cR_{\bx}\hspace{-2mm}&\triangleq\hspace{-2mm}&\expec\{\cR_{\bx}(i)\}\hspace{-0.5mm}=\diag\Big\{\hspace{-1.5mm}\sum_{\ell\in\N_k\cap\C(k)}\hspace{-5mm}c_{\ell k}\bR_{\bx,\ell}\Big\}_{k=1}^N\\
\expec\{\br(i+1)\}\hspace{-2mm}&\triangleq\hspace{-2mm}&\eta\cM\expec\{\bGamma_{i+1}(\bphi(i+1))\}.\label{eq: expec r(i+1)}
\end{eqnarray}
The following theorem guarantees the mean stability of the multitask diffusion LMS~\eqref{eq: ATC FBS} with forward-backward splitting.

Recall that the block maximum norm of an $N\times 1$ block vector $\bx=\col\{\bx_k\}_{k=1}^N$ and the induced block maximum norm of an $N\times N$ block matrix $\cX$ are defined as~\cite{sayed2014diffusion}:
\begin{equation}
	\begin{split}
	&\|\bx\|_{b,\infty}=\max\limits_{1\leq k \leq N} \|\bx_k\|_2,\\
	&\|\cX\|_{b,\infty}=\max\limits_{\bx} \frac{\|\cX\bx\|_{b,\infty}}{\|\bx\|_{b,\infty}},
	\end{split}
\end{equation}

\begin{theorem}
\label{theorem: stability in the mean}
\textbf{\emph{(Stability in the mean)}} Assume data model~\eqref{eq: linear data model} and Assumption \ref{assumption: independent regressors} hold. Then, for any initial conditions, the multitask diffusion strategy~\eqref{eq: ATC FBS} converges in the mean to a small bounded region of the order of $\mu_{\max}$, i.e., $\lim_{i\rightarrow\infty} \expec\{\|\bwt(i)\|_{b,\infty}\}\;=\;O(\mu_{\max})$, if the step-sizes are chosen such that:
\begin{equation}
	\label{eq: step-size condition}
	0<\mu_k<\frac{2}{\lambda_{\max}(\sum_{\ell\in\N_k\cap\C(k)}c_{\ell k}\bR_{\bx,\ell})},\quad k=1,\ldots,N,
\end{equation}
where $\mu_{\max}\triangleq\max_{1\leq k\leq N} \mu_k$ and $\lambda_{\max}(\cdot)$ is the maximum eigenvalue of its matrix argument. The block maximum norm of the bias can be upper bounded as:
\begin{eqnarray}
	\lim_{i\rightarrow\infty}\|\expec\{\bwt(i)\}\|_{b,\infty}&\leq& 
	\frac{\eta\,\mu_{\max}\,s_{\max}\sqrt{M}}{1-\|\cB\|_{b,\infty}},\label{eq: bound l1 2}\\
	\lim_{i\rightarrow\infty}\|\expec\{\bwt(i)\}\|_{b,\infty}& \leq&
	\frac{1}{\epsilon}\cdot\frac{\eta\,\mu_{\max}\,s_{\max}\sqrt{M}}{1-\|\cB\|_{b,\infty}},\label{eq: bound rew l1 2}
\end{eqnarray}
for the $\ell_1$-norm and the reweighted $\ell_1$-norm, respectively.

\end{theorem}

\begin{proof}
Iterating \eqref{eq: mean error recursion} starting from $i=0$, we arrive to the following expression:
\begin{equation}
	\label{eq: mean error recursion 2}
	\expec\{\bwt(i+1)\}=
	\cB^{i+1}\expec\{\bwt(0)\}+\sum_{j=0}^i \cB^j\,\expec\{\br(i+1-j)\},
\end{equation}
where $\expec\{\bwt(0)\}$ is the initial condition. $\expec\{\bwt(i+1)\}$ converges when $i\rightarrow\infty$ if, and only if, both terms on the RHS of~\eqref{eq: mean error recursion 2} converges to finite values. The first term converges to zero as $i\rightarrow \infty$ if the matrix $\cB$ is stable. A sufficient condition to ensure the stability of $\cB$ is to choose the step-sizes according to~\eqref{eq: step-size condition} (the proof can be obtained using the same arguments as~\cite[Theorem 5.1]{sayed2014diffusion}). We shall now prove the convergence of the second term on the RHS of~\eqref{eq: mean error recursion 2}. To prove the convergence of the series $\sum_{j=0}^{+\infty} \cB^j\,\expec\{\br(i+1-j)\}$, it is sufficient to prove that the series $\sum_{j=0}^{+\infty} [\cB^j\,\expec\{\br(i+1-j)\}]_k$ converges for $k=1,\ldots,MN$. A series is absolutely convergent if each term of the series can be bounded by a term of an absolutely convergent series~\cite{Lorenzo2013sparse}. Since the block maximum norm of a block vector is greater than or equal to the largest absolute value of its entries, each term $\big|[\cB^j\,\expec\{\br(i+1-j)\}]_k\big|$ can be bounded as:
\begin{align}
	\big|[\cB^j\,\expec\{\br(i+1-j)\}]_k\big|&\leq \|\cB\|_{b,\infty}^j   \cdot\|\expec\{\br(i+1-j)\}\|_{b,\infty}\nonumber\\
	&\leq \|\cB\|_{b,\infty}^j   r_{\max}.
\end{align}
The quantity $\|\expec\{\br(i+1-j)\}\|_{b,\infty}$ is finite for all $i$ and $j$ and bounded by some constant $r_{\max}=\cO(\mu_{\max})$. In fact, from~\eqref{eq: expec r(i+1)}, we have:
\begin{equation}
\label{eq: bound on r}
	\|\expec\{\br(i+1)\}\|_{b,\infty}\leq\eta\mu_{\text{max}}\|\expec\{\bGamma_{i+1}(\bphi(i+1))\}\|_{b,\infty}
\end{equation}
since $\|\cM\|_{b,\infty}=\mu_{\text{max}}$. Using~\eqref{eq: bound l1}--\eqref{eq: bound rew l1}, the block maximum norm of $\bGamma_{i+1}(\bphi(i+1))$ in~\eqref{eq: Gamma i} can be bounded as:
\begin{eqnarray}
	\|\bGamma_{i+1}(\bphi(i+1))\|_{b,\infty}\hspace{-2.5mm}&\leq\hspace{-2.5mm}& s_{\text{max}}\sqrt{M},
	(\ell_1\text{-norm})\label{eq: bound l1 2}\\
	\|\bGamma_{i+1}(\bphi(i+1))\|_{b,\infty}\hspace{-2.5mm}& \leq\hspace{-2.5mm}& \frac{s_{\text{max}}\sqrt{M}}{\epsilon},
	(\text{rew. } \ell_1\text{-norm})\label{eq: bound rew l1 2}
\end{eqnarray}
for all $i$, where $s_\text{max}=\max\limits_{1\leq k\leq N} s_k$. If the step-sizes are chosen according to~\eqref{eq: step-size condition}, the series $\sum_{j=0}^{+\infty} \|\cB\|_{b,\infty}^j r_{\text{max}}$ is absolutely convergent. Therefore, the series $\sum_{j=0}^{+\infty} [\cB^j\,\expec\{\br(i+1-j)\}]_k$ is an absolutely convergent series.

Note that when $i\rightarrow\infty$, the block maximum norm of the bias can be bounded as
\begin{align}
	\lim_{i\rightarrow\infty}\|\expec\{\bwt(i)\}\|_{b,\infty}
	&=\lim_{i\rightarrow\infty}\Big\|\sum_{j=0}^i \cB^j\,\expec\{\br(i+1-j)\}\Big\|_{b,\infty}\nonumber\\
	&\leq\lim_{i\rightarrow\infty}\sum_{j=0}^\infty\|\cB^j\,\expec\{\br(i+1-j)\}\|_{b,\infty}\nonumber\\
	&\leq\lim_{i\rightarrow\infty}\sum_{j=0}^\infty\|\cB\|_{b,\infty}^j r_{\text{max}}=\frac{r_{\text{max}}}{1-\|\cB\|_{b,\infty}},
\end{align}
\end{proof}


\subsection{Mean-square-error stability}

We examine the mean-square-error stability by studying the convergence of the weighted variance $\expec\{\|\bwt(i)\|_{\bSig}^2\}$, where $\bSig$ is a positive semi-definite matrix that we are free to choose. Evaluating the variance, we obtain:
\begin{equation}
\begin{split}
	\label{eq: variance relation 1}
	\expec\{\|\bwt(i+1)\|^2_{\bSig}\}=&\expec\{\|\bwt(i)\|^2_{\bSig'}\}+\expec\{\|\bg(i)\|^2_{\bSig}\}+\\
	&\varphi(\br(i+1),
	\bSig,\cB(i),\bwt(i),\bg(i)),
\end{split}
\end{equation}
where $\bSig'\triangleq\expec\{\cB^\top(i)\bSig\cB(i)\}$ and 
\begin{equation}
\begin{split}
	\label{eq: term resulting from cooperation 1}
	\varphi(\br&(i+1),\cB(i),\bwt(i),\bg(i))	=\expec\{\|\br(i+1)\|^2_{\bSig}\}+\\
	&2\expec\{\br^\top(i+1)\bSig\cB(i)\bwt(i)\}-2\expec\{\br^\top(i+1)\bSig\bg(i)\}
\end{split}
\end{equation}
is a term coming from promoting relationships between clusters. The last two terms on the RHS of \eqref{eq: term resulting from cooperation 1} contain higher-order powers of the step-sizes. Using Assumption~\ref{assumption: step-sizes assumption}, we get the following approximation:
\begin{equation}
\label{eq: term resulting from cooperation}
\varphi(\br(i+1),\bwt(i))\approx\expec\{\|\br(i+1)\|^2_{\bSig}\}+2\expec\{\br^\top(i+1)\bSig\cB\bwt(i)\}
\end{equation}
Let $\bsig\triangleq\vc(\bSig)$ and $\bsig' \triangleq\vc(\bSig')$ where the $\vc(\cdot)$ operator stacks the columns of a matrix on top of each other. We will use the notation $\|\bwt\|_{\bsig}^2$ and $\|\bwt\|_{\bSig}^2$ interchangeably to denote the same quantity $\bwt^\top\bSig\bwt$. Using the property $\vc(\bU\bSig\bW)=(\bW^\top\otimes\bU)\vc(\bSig)$, the relation between $\bsig'$ and $\bsig$ can be expressed in the following form:
\begin{equation}
	\bsig'=\cF\bsig,
\end{equation}
where $\cF$ is the $(LN)^2\times(LN)^2$ matrix given by:
\begin{equation}
	\label{eq: matrix F}
	\cF\triangleq\expec\{\cB^\top(i)\otimes\cB^\top(i)\}\approx\cB^\top\otimes\cB^\top.
\end{equation}
The approximation in~\eqref{eq: matrix F} is reasonable under Assumption~\ref{assumption: step-sizes assumption}\cite{sayed2014diffusion}. Introducing the matrix $\bG$:
\begin{equation}
	\bG\triangleq\expec\{\bg(i)\bg^\top(i)\}=\cAT\cM\,\cCT\diag\{\bR_{\bx,k}\sigma^2_{z,k}\}_{k=1}^N\cC\cM\cA
\end{equation}
and using the property $\tr(\bSig\bX)=[\vc(\bX^\top)]^\top\vc(\bSig)$, the second term on the RHS of~\eqref{eq: variance relation 1} can be written as:
\begin{equation}
	\expec\{\|\bg(i)\|^2_{\bSig}\}=[\vc(\bG^\top)]^\top\bsig.
\end{equation}
Hence, the variance recursion~\eqref{eq: variance relation 1} can be expressed as
\begin{equation}
\begin{split}
	\label{eq: variance relation 2} 
	\expec\{\|\bwt(i+1)\|^2_{\bsig}\}=\expec\{&\|\bwt(i)\|^2_{\cF\bsig}\}+[\vc(\bG^\top)]^\top\bsig+\\
	&\varphi(\br(i+1),\bsig,\bwt(i)).
\end{split}
\end{equation}

\begin{theorem}
\label{theorem: stability in the mean-square}
\textbf{\emph{(Mean-square-error Stability)}}{ Assume data model~\eqref{eq: linear data model} and Assumptions~\ref{assumption: independent regressors} and~\ref{assumption: step-sizes assumption} hold. Then, for any initial conditions, the multitask diffusion strategy~\eqref{eq: ATC FBS} is mean-square stable if the error recursion~\eqref{eq: network error vector recursion} is mean stable and the matrix $\cF$ is stable. Using the approximation~\eqref{eq: matrix F}, the matrix $\cF$ is stable if the step-sizes satisfy~\eqref{eq: step-size condition}.}
\end{theorem}

\begin{proof}
Since $\bSig$ is a positive semi-definite matrix and the vector $\br(i+1)$ is uniformly bounded for all $i$, $\expec\{\|\br(i+1)\|^2_{\bSig}\}$ can be bounded as
\begin{equation}
	\label{eq: bound on term 1 of f}
	0\leq\expec\{\|\br(i+1)\|^2_{\bSig}\}\leq \kappa_1
\end{equation}
for all $i$, where $\kappa_1$ is a positive constant. Since $\br(i+1)$ is uniformly bounded for all $i$, the vector $2\br^\top(i+1)\bSig\,\cB$ is also bounded for all $i$. Let $\gamma_{\text{max}}$ be a bound on the largest component of $2\br^\top(i+1)\bSig\,\cB$ in absolute value for all $i$. We obtain
\begin{align}
	2|\expec\{\br^\top(i+1)\bSig\cB\bwt(i)\}|&\leq \gamma_{\max}\nonumber
	\sum_{\ell=1}^{MN}\big| \expec\big\{\bwt_\ell(i)\big\}\big|\\
	&=\gamma_{\max}\cdot
	\| \expec\big\{\bwt(i)\big\}\|_1.
\end{align}
Under condition \eqref{eq: step-size condition} on the step-sizes, the mean error vector $\expec\{\bwt(i)\}$ converges to a small bounded region as $i\rightarrow\infty$. Hence, $\| \expec\{\bwt(i)\}\|_1$ can be upper bounded by some positive constant scalar $\kappa_2$ for all $i$, and using the approximation~\eqref{eq: term resulting from cooperation},  $|\varphi(\br(i+1),\bsig,\bwt(i))|$ satisfies:
\begin{equation}
	\label{eq: bound on f}
	|\varphi(\br(i+1),\bsig,\bwt(i))|\leq \kappa_1+\gamma_{\max}\kappa_2
\end{equation}
for all $i$. The positive constant $\kappa_3\triangleq \kappa_1+\gamma_{\max}\kappa_2$ can be written as a scaled multiple of the positive quantity $[\vc(\bG^\top)]^\top\bsig$ as $\kappa_3=t[\vc(\bG^\top)]^\top\bsig$ where $t\geq 0$~\cite{Lorenzo2013sparse}. We arrive at the following inequality for \eqref{eq: variance relation 2}: 
\begin{equation}
	\label{eq: variance relation 3} 
	\expec\{\|\bwt(i+1)\|^2_{\bsig}\}\leq\expec\{\|\bwt(i)\|^2_{\cF\bsig}\}+(1+t)\cdot[\vc(\bG^\top)]^\top\bsig.
\end{equation}
Iterating~\eqref{eq: variance relation 3} starting from $i=0$, we obtain
\begin{equation}
\begin{split}
	\label{eq: variance relation 4} 
	&\expec\{\|\bwt(i+1)\|^2_{\bsig}\}\\
	&\leq\expec\{\|\bwt(0)\|^2_{\cF^{i+1}\bsig}\}
	+(1+t)[\vc(\bG^\top)]^\top\sum_{j=0}^i\cF^j\bsig,
\end{split}
\end{equation} 
where $\expec\{\|\bwt(0)\|^2\}$ is the initial condition. If we show that the RHS of \eqref{eq: variance relation 4} converges, then $\expec\{\|\bwt(i+1)\|^2_{\bsig}\}$ is stable. The first term on the RHS of~\eqref{eq: variance relation 4} vanishes as $i\rightarrow\infty$ if the matrix $\cF$ is stable. Consider now the second term on the RHS of~\eqref{eq: variance relation 4}. The series $\sum_{j=0}^\infty \cF^j\bsig$ converges if $\sum_{j=0}^\infty[ \cF^j\bsig]_k$ converges for $k=1,\ldots,(MN)^2$. Each term of the series can be bounded as
\begin{equation}
	\label{eq: equation 1}
	[\cF^j\bsig]_k\leq|[\cF^j\bsig]_k|\leq\|\cF^j\bsig\|_{b,\infty}\leq\|\cF^j\|_{b,\infty}\cdot\|\bsig\|_{b,\infty}.
\end{equation}
Since $\cF$ is stable, there exists a submultiplicative norm\footnote{
The norm $\|\cdot\|_{\rho}$ is called submultiplicative if for any square matrices $\bU$ and $\bW$ of compatible dimensions we have: $\|\bU\bW\|_{\rho}\leq\|\bU\|_{\rho}\cdot\|\bW\|_{\rho}$.}
$\|\cdot\|_{\rho}$ such that $\|\cF\|_{\rho}=\zeta<1$. All norms are equivalent in finite dimensional vector spaces. Thus, we have: 
\begin{equation}
	\|\cF^j\|_{b,\infty}\leq \tau\|\cF^j\|_{\rho}\leq\tau\|\cF\|^j_{\rho}=\tau\zeta^j,
\end{equation}
for some positive constant $\tau$. Considering this bound with~\eqref{eq: equation 1} yields:
\begin{equation}
\begin{split}
	\sum_{j=0}^\infty|[ \cF^j\bsig]_k|\leq\sum_{j=0}^\infty\|\cF^j\|_{b,\infty}\cdot\|\bsig\|_{b,\infty}
	&\leq\tau\sum_{j=0}^\infty\zeta^j\|\bsig\|_{b,\infty}\\
	&=\frac{\tau\cdot\|\bsig\|_{b,\infty}}{1-\zeta}.
\end{split}
\end{equation}
As a consequence, since the second term on the RHS of~\eqref{eq: variance relation 4} converges to a bounded region when $\cF$ is stable, $\expec\{\|\bwt(i+1)\|^2_{\bsig}\}$ also converges. 
\end{proof}


\section{Simulation results}
\label{sec: simulation results}
Before proceeding, we present a new rule for selecting the regularization weight $p_{k\ell}$ based on a measure of sparsity of the vector $\bw_k^o-\bw_\ell^o$. The intuition behind this rule is to employ a large weight $p_{k\ell}$ when the objectives at nodes $k$ and $\ell$ have few distinct entries, i.e., $\bw_k^o-\bw_{\ell}^o$ is sparse, and a small weight $p_{k\ell}$ when the objectives have few similar entries, i.e., $\bw_k^o-\bw_{\ell}^o$ is not sparse. Among other possible choices for the sparsity measure, we select a popular one based on a relationship between the $\ell_1$-norm and $\ell_2$-norm~\cite{hoyer2004non-negative}: 
\begin{equation}
\xi(\bw_k^o-\bw_\ell^o)=\frac{M}{M-\sqrt{M}}\Big(1-\frac{\|\bw_k^o-\bw_\ell^o\|_1}{\sqrt{M}\cdot\|\bw_k^o-\bw_\ell^o\|_2}\Big)~\in[0,1].
\end{equation}
The quantity $\xi(\bw_k^o-\bw_\ell^o)$ is equal to one when only a single component of $\bw_k^o-\bw_\ell^o$ is non-zero, and zero when all elements of  $\bw_k^o-\bw_\ell^o$ are relatively large~\cite{hoyer2004non-negative}. Since the nodes do not know the true objectives $\bw^o_k$ and $\bw^o_{\ell}$, we propose to replace these quantities by the available estimates at each time instant $i$ and allow the regularization factors to vary with time according to:
\begin{equation}
\label{eq: adaptive regularization factors}
p_{k\ell}(i)\propto
\left\lbrace
\begin{array}{lr}
\frac{M}{M-\sqrt{M}}\Big(1-\frac{\|\bphi_k(i+1)-\bphi_{\ell}(i+1)\|_1}{\sqrt{M}\cdot\|\bphi_k(i+1)-\bphi_{\ell}(i+1)\|_2}\Big),\\
\qquad\qquad\qquad\qquad\text{if } \ell\in{\N_k\setminus\C(k)}\\
0,\quad\qquad\qquad\qquad \text{otherwise}
\end{array}
\right.
\end{equation}
where the symbol $\propto$ denotes proportionality. As we shall see in the simulations, this rule improves the performance of the algorithm and allows agent $k$ to adapt the regularization strength $p_{k\ell}$ with respect to the sparsity level of the vector $\bw^o_k-\bw^o_{\ell}$ at time instant $i$.
 
\subsection{Illustrative example}
\label{subsec: Illustrative example}

We consider a clustered network with the topology shown in Fig.~\ref{fig: setup}(a), consisting of $20$ nodes divided into $3$ clusters: $\C_1=\{1,\ldots,10\}$, $\C_2=\{11,\ldots,15\}$, and $\C_3=\{16,\ldots,20\}$. The regression vectors $\bx_k(i)$ are $18\times1$ zero-mean Gaussian distributed vectors with covariance matrices $\bR_{\bx,k}=\sigma^2_{x,k}\bI_{18}$. The variances $\sigma^2_{x,k}$ are shown in Fig.~\ref{fig: setup}(b). The noises $z_k(i)$ are zero-mean i.i.d. Gaussian random variables, independent of any other signal, with variances $\sigma^2_{z,k}$ shown in Fig.~\ref{fig: setup}(b). 
\begin{figure}
	\centering
	\subfigure[Network topology.]{\includegraphics[trim = 10mm 10mm 10mm 11mm, clip, scale=0.25]{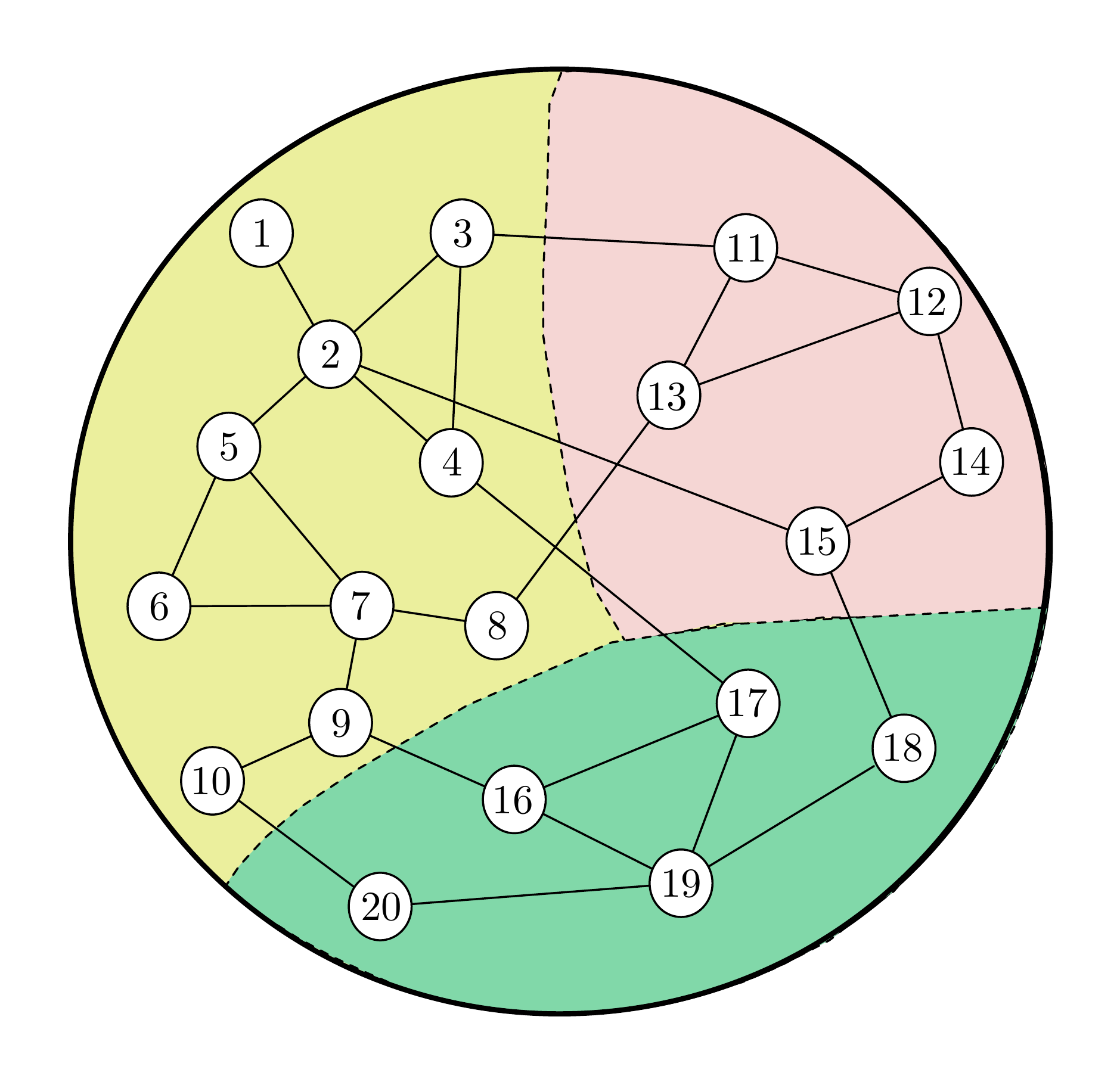}} 
	\subfigure[Regression and noise variances.]{ \includegraphics[trim = 13mm 75mm 23mm 80mm, clip,
	scale=0.25]{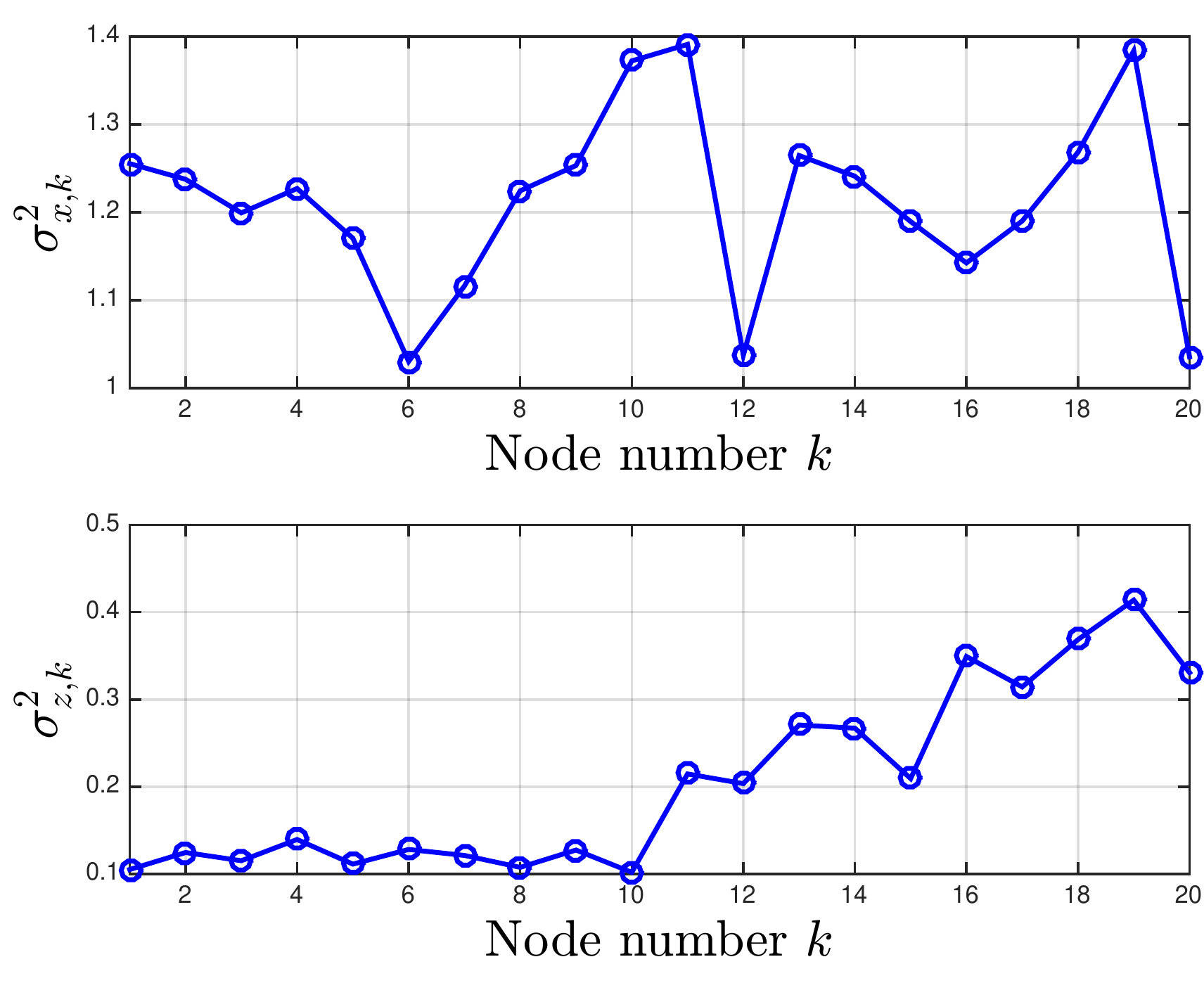}} 
	\caption{Experimental setup.}
	\label{fig: setup}
\end{figure}
Let $\card\{\cdot\}$ denote the cardinal of its entry. We run the diffusion algorithm~\eqref{eq: ATC FBS} by setting $c_{\ell k}=\frac{1}{\card\{\N_{\ell}\cap\C(\ell)\}}$ for $k\in\N_{\ell}\cap\C(\ell)$ and $a_{\ell k}=\frac{1}{\card\{\N_{k}\cap\C(k)\}}$ for $\ell\in\N_{k}\cap\C(k)$. The regularization weights are set to $\rho_{k\ell}=\frac{1}{\card\{\N_k\setminus\C(k)\}}$ for $\ell\in \N_k\setminus\C(k)$. We use a constant step-size $\mu=0.02$ for all nodes, a sparsity strength $\eta=0.06$ for the $\ell_1$-norm regularizer, and $\eta=0.04$ for the reweighted $\ell_1$-norm regularizer with $\epsilon=0.1$ . The results are averaged over $200$ Monte-Carlo runs.

The optimum vectors are set to $\bw^o_{\C_j}=\bw^o\,+\,\bdelta_{\C_j}$ at each cluster with $\bw^o$ an $18\times 1$ vector whose entries are generated from the Gaussian distribution $\N(0,1)$. First, we set $\bdelta_{\C_1}$ to $\boldsymbol{0}_{1\times18}^\top$, $\bdelta_{\C_2}$ to $[-1~\boldsymbol{0}_{1\times17}]^\top$, and $\bdelta_{\C_3}$ to $[\boldsymbol{0}_{1\times6}~-1~\boldsymbol{0}_{1\times11}]^\top$. Observe that at most two entries differ between clusters. After $500$ iterations, we set $\bdelta_{\C_2}$ to $[-\cb{1}_{1\times3}~1~\boldsymbol{0}_{1\times14}]^\top$ and $\bdelta_{\C_3}$ to $[\cb{0}_{1\times12}~-\cb{1}_{1\times3}~\boldsymbol{0}_{1\times3}]^\top$. In this way, at most $7$ entries differ between clusters. After $1000$ iterations, we set $\bdelta_{\C_2}$ to $[-\cb{1}_{1\times3}~\cb{1}_{1\times3}~-\cb{1}_{1\times3}~\boldsymbol{0}_{1\times9}]^\top$ and $\bdelta_{\C_3}$ to $[\boldsymbol{0}_{1\times9}~\cb{1}_{1\times3}~\cb{-1}_{1\times3}~\cb{1}_{1\times3}]^\top$. Thus, at most $18$ entries now differ between clusters.

In Fig.~\ref{fig: network MSD 6 algorithms}, we compare $6$ algorithms:  the non-cooperative LMS (algorithm~\eqref{eq: ATC FBS} with $\bA=\bC=\bI_N$ and $\eta=0$), the regularized LMS (algorithm~\eqref{eq: ATC FBS} with $\bA=\bC=\bI_N$) with $\ell_1$-norm and reweighted $\ell_1$-norm, the multitask diffusion LMS without regularization (algorithm~\eqref{eq: ATC FBS} with $\eta=0$), and the multitask diffusion LMS~\eqref{eq: ATC FBS} with $\ell_1$-norm and reweighted $\ell_1$-norm regularization.  As observed in this figure, when the tasks share a sufficient number of components, cooperation between clusters enhances the network MSD performance. When the number of common entries decreases, the cooperation between clusters becomes less effective. The use of the $\ell_1$-norm can lead to a degradation of the MSD relative to the absence of cooperation among clusters. However, the use of the reweighted $\ell_1$-norm allows to improve the performance.

In order to better understand the behavior of the algorithm~\eqref{eq: ATC FBS} in the clusters, we report in Fig.~\ref{fig: common and different components} the learning curves for $i\in[0,1000]$ of the common and distinct entries among clusters given by 
\begin{equation}
	\frac{1}{\card\{\C_j\}}\sum_{k\in\C_j}\expec\Big\{\sum_{m\in\Omega(i)}
	([\bw^o_k(i)-\bw_k(i)]_m)^2\Big\},
\end{equation} 
for $j=1$, $3$, where $\Omega(i)$ is the set of identical (or distinct) components among all clusters at iteration $i$ and $\bw^{o}_k(i)$ is the optimum parameter vector at node $k$ and iteration $i$. For example, for $i\in[0,500]$, the set of distinct components is $\{1,7\}$. As shown in this figure, cluster $\C_3$ benefits considerably from cooperation with other clusters in the estimation of the common entries. Nevertheless, cluster $\C_1$ benefits slightly from cooperation. This is due to the fact that the performance of $\C_3$ is low relatively to that of $\C_1$ since the SNR in $\C_3$ is small and the number of nodes employed in this cluster is $5$.
\begin{figure}
\centering
\includegraphics[trim = 0mm 0mm 0mm 0mm, clip, scale=0.38]{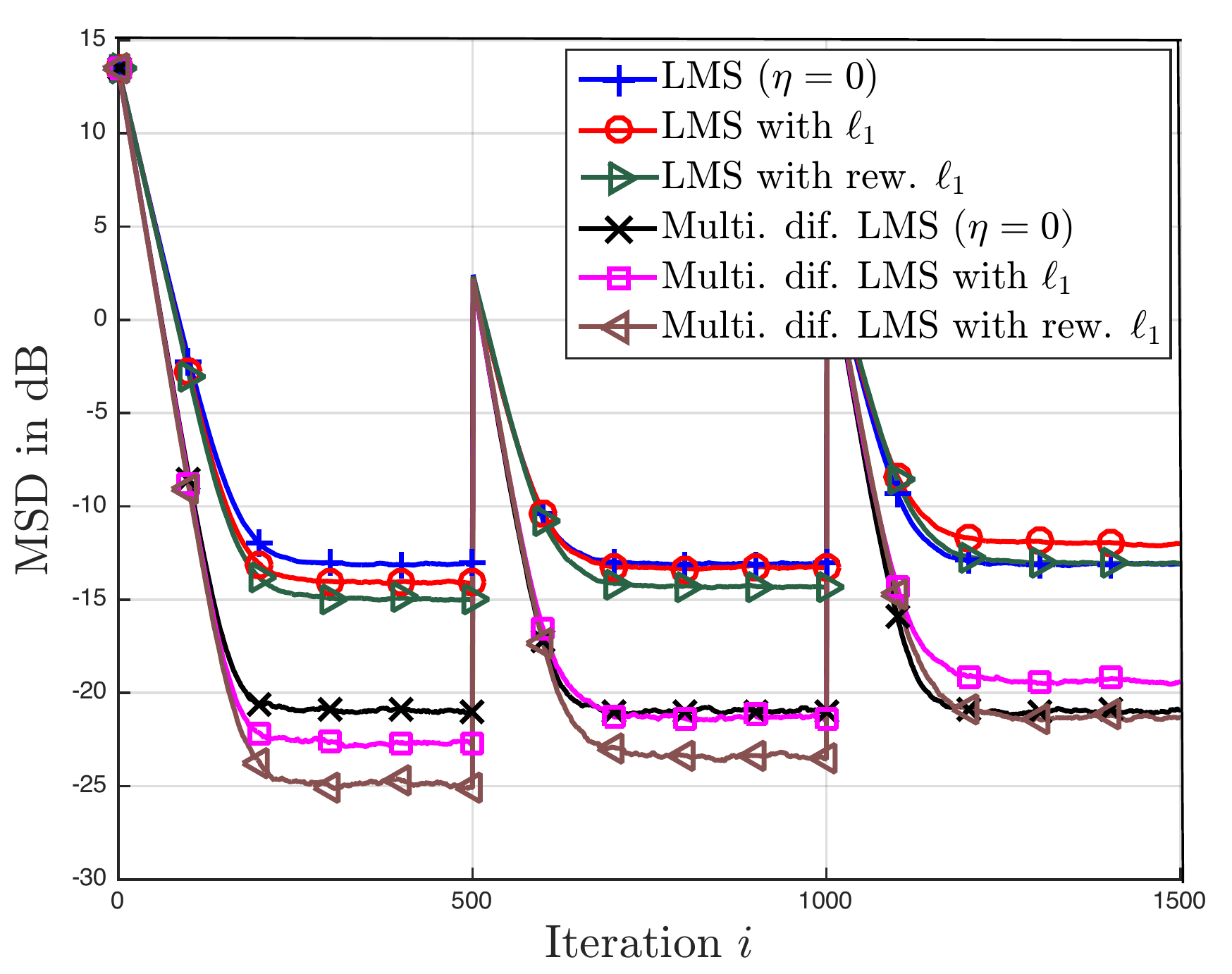}
\caption{Network MSD comparison for $6$ different strategies: non-cooperative LMS (algorithm~\eqref{eq: ATC FBS} with $\bA=\bC=\bI_N$, $\eta=0$), spatially regularized LMS (algorithm~\eqref{eq: ATC FBS} with $\bA=\bC=\bI_N$ with $\ell_1$-norm and reweighted $\ell_1$-norm, standard diffusion without cooperation between clusters (algorithm~\eqref{eq: ATC FBS} with $\eta=0$), and our proximal diffusion~\eqref{eq: ATC FBS} with $\ell_1$-norm and reweighted $\ell_1$-norm.}
\label{fig: network MSD 6 algorithms}
\end{figure}

\begin{figure*}[t]
\centering
\subfigure[Cluster $1$ MSD over identical entries.]{\includegraphics[trim = 15mm 70mm 15mm 80mm, clip, scale=0.33]{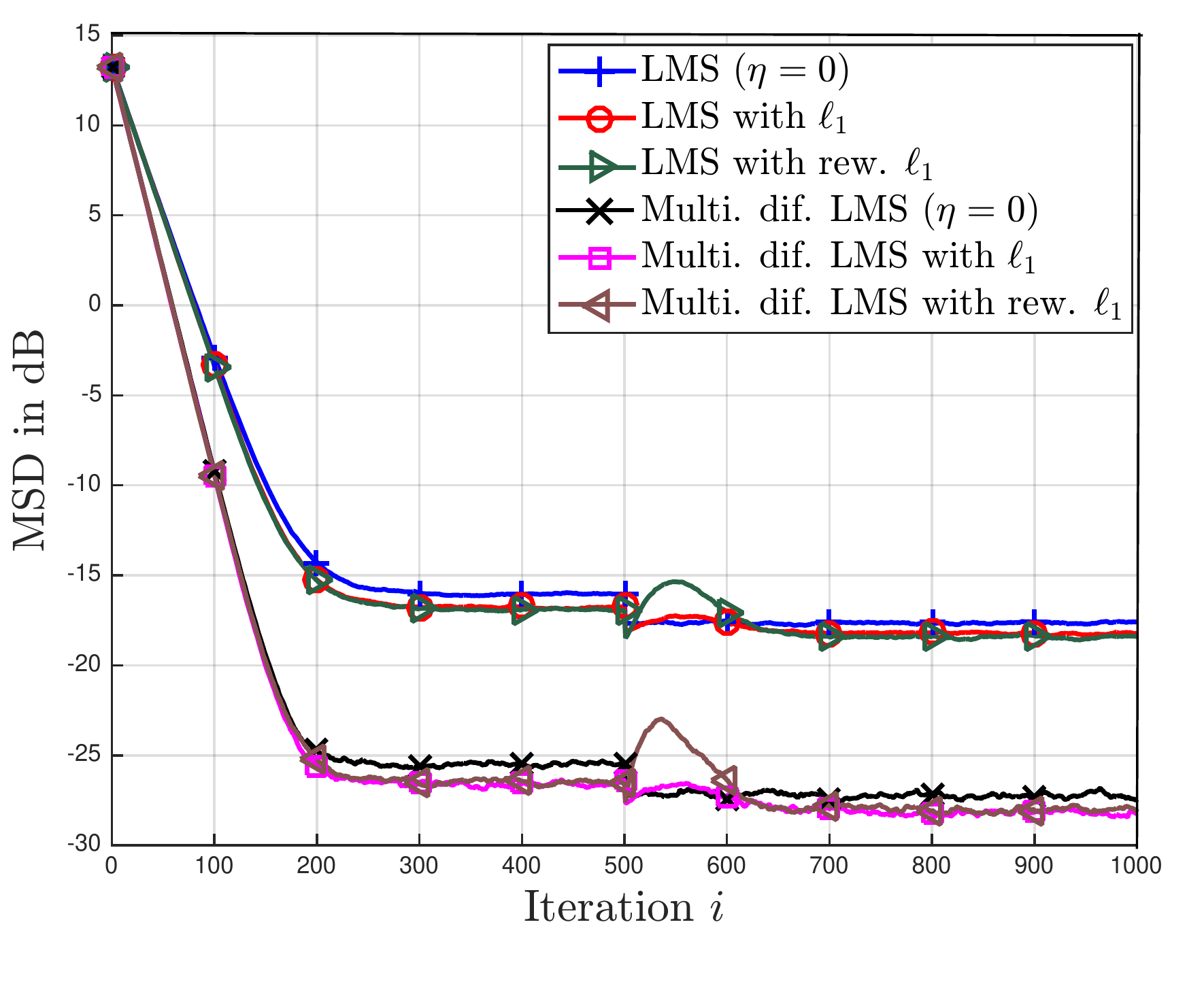}} \qquad
\subfigure[Cluster $3$ MSD over identical entries.]{ \includegraphics[trim = 15mm 70mm 15mm 80mm, clip, scale=0.33]{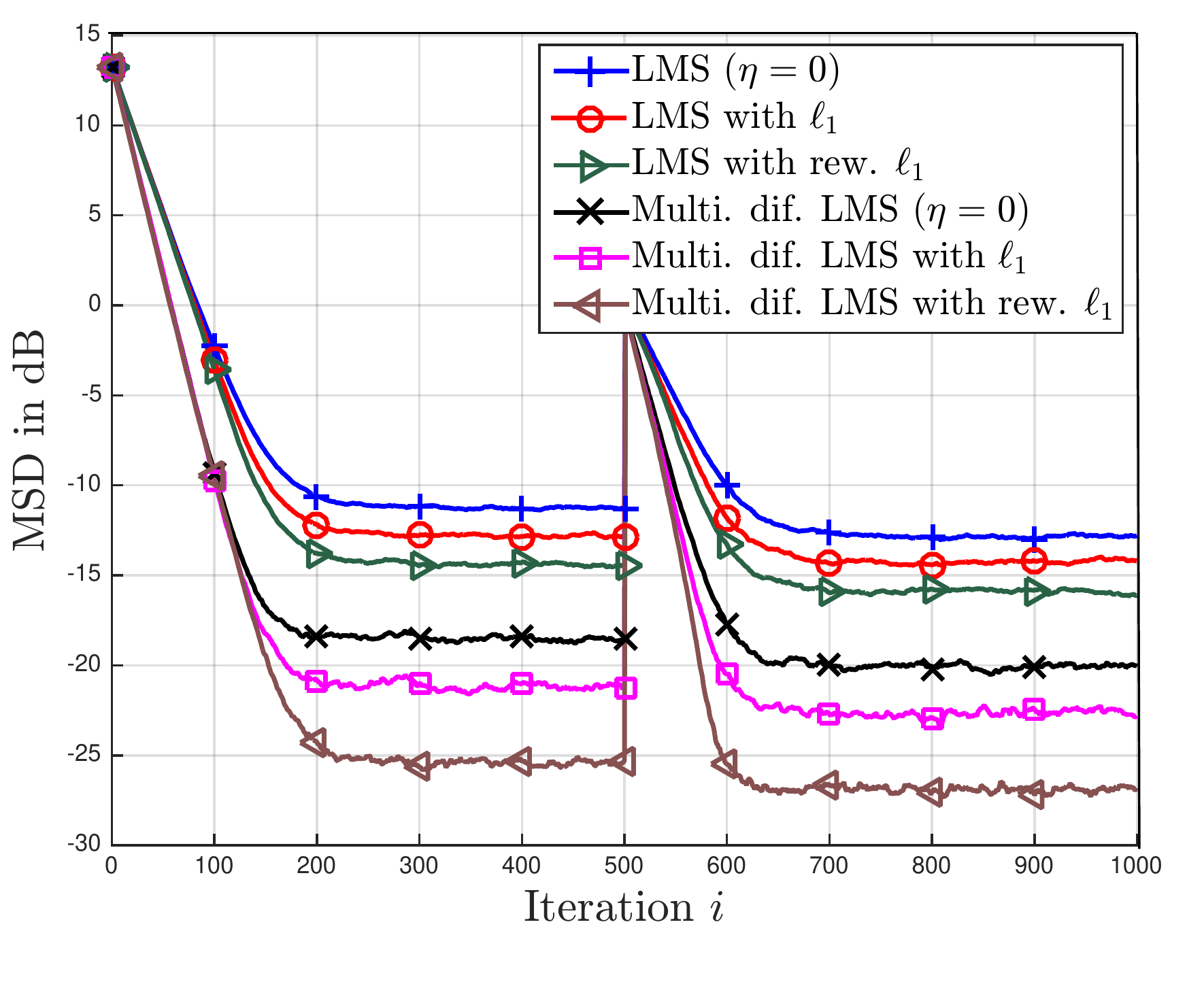}} \\
\subfigure[Cluster $1$ MSD over distinct entries.]{ \includegraphics[trim = 15mm 70mm 15mm 80mm, clip, scale=0.33]{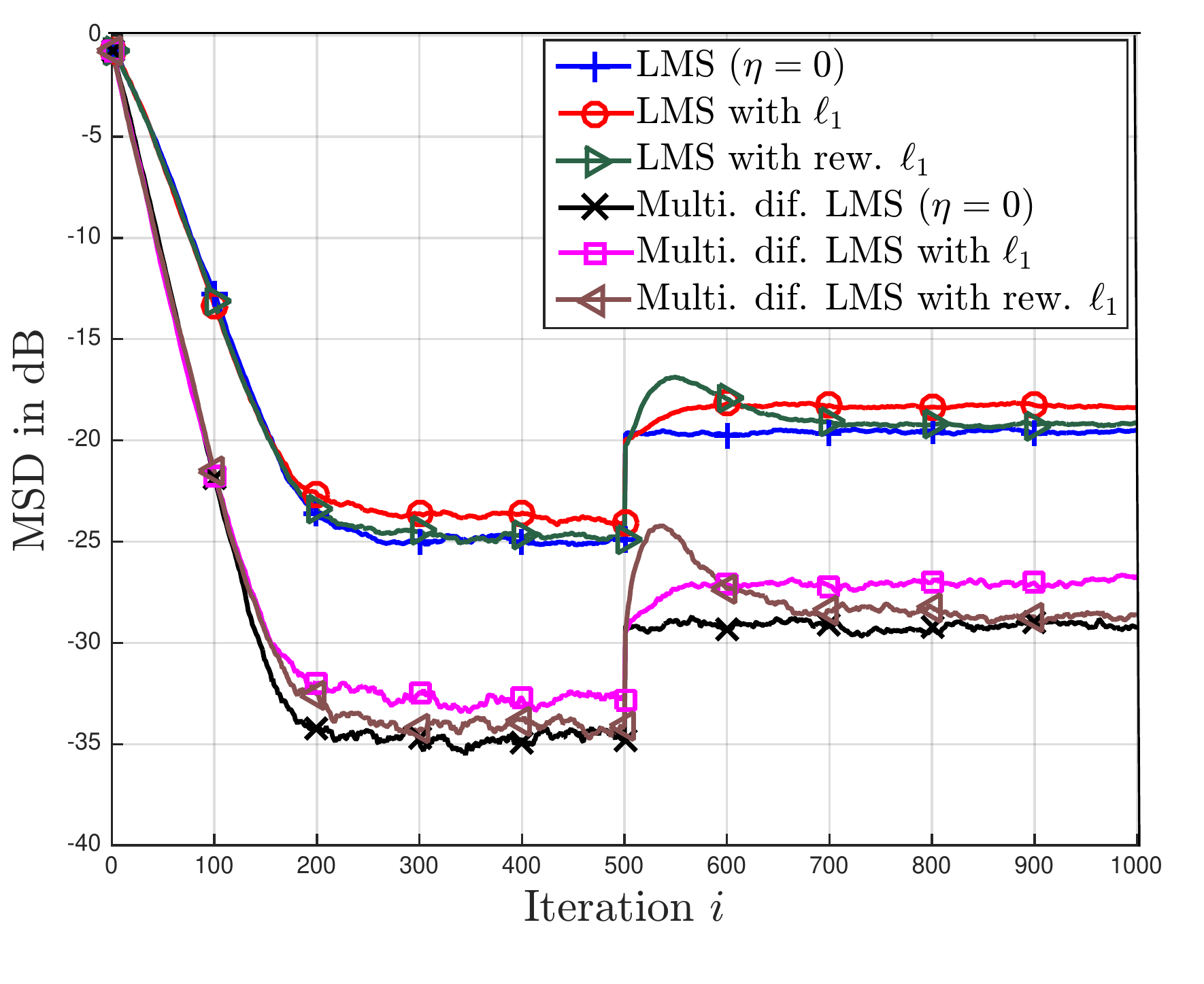}} \qquad
\subfigure[Cluster $3$ MSD over distinct entries.]{ \includegraphics[trim = 15mm 70mm 15mm 80mm, clip, scale=0.33]{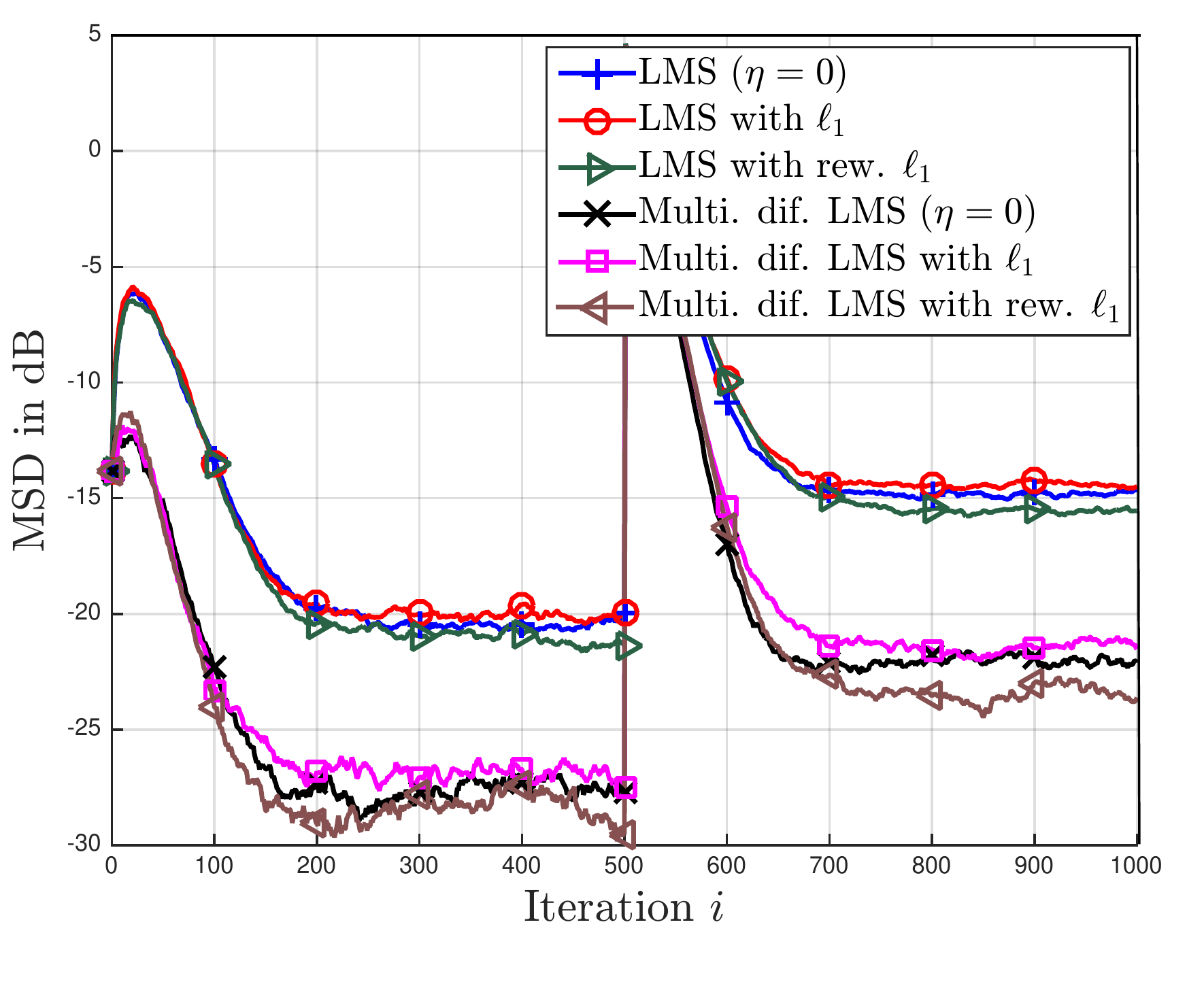}} 
\caption{Clusters MSD over identical and distinct components. Comparison for the same $6$ different strategies considered in Fig.~\ref{fig: network MSD 6 algorithms}.}
\label{fig: common and different components}
\end{figure*}

We shall now illustrate the effect of the regularization strength $\eta$ over the performance of the algorithm for different numbers of common entries between the optimum vectors $\bw^o_k$. We consider the same settings as above, which means that the number of common entries among clusters is successively set to $16$, $11$, and $0$ over $18$. Parameter $\eta$ is uniformly sampled over $[0,0.14]$. Figure~\ref{fig: effect of regularization} shows the gain in steady-state MSD versus the unregularized algorithm obtained for $\eta=0$, as a function of $\eta$. For each $\eta$, the results are averaged over $50$ Monte-Carlo runs and over $50$ samples after convergence of the algorithm. It can be observed in Fig.~\ref{fig: effect of regularization} that the interval for $\eta$ over which the network benefits from cooperation between clusters becomes smaller as the number of common entries decreases. In addition, the reweighted $\ell_1$-norm regularizer provides better performance than the $\ell_1$-norm regularizer.

\begin{figure*}[t]
	\centering
	\subfigure[ $\ell_1$-norm.]{\includegraphics[trim = 0mm 0mm 0mm 0mm, clip, scale=0.35]
	{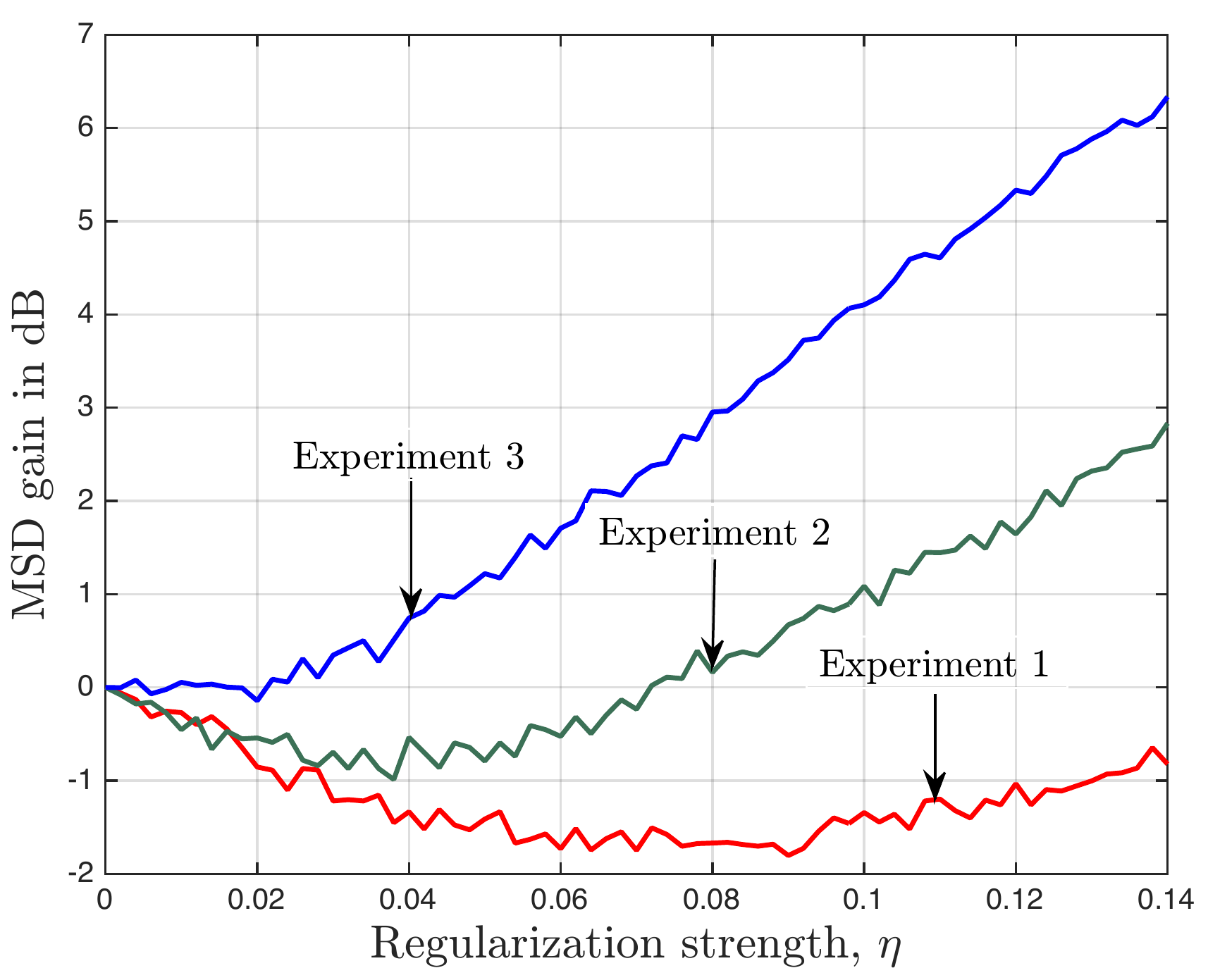}}
	\subfigure[Reweighted $\ell_1$-norm.]{ \includegraphics[trim = 0mm 0mm 0mm 0mm, clip, scale=0.35]
	{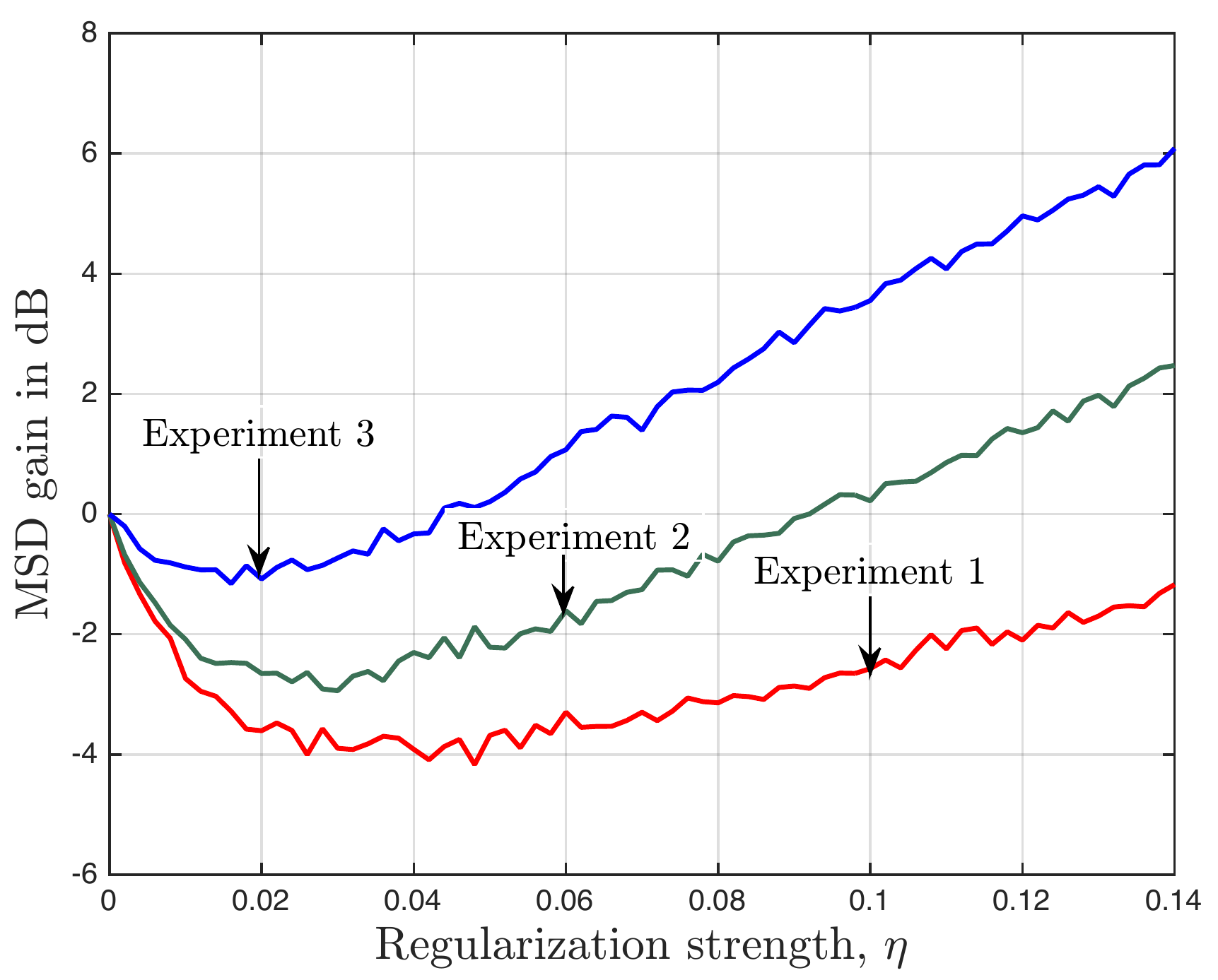}} 
	\caption{Differential network MSD ($\text{MSD}(\eta)-\text{MSD}(\eta=0)$) in dB with respect to the regularization strength $\eta$ for the multitask diffusion
	LMS~\eqref{eq: ATC FBS}  with $\ell_1$-norm (left) and reweighted $\ell_1$-norm (right) for $3$ different degrees of similarity between tasks. Experiment 1: at most 2 entries differ between clusters. Experiment 2: at most 7 entries differ between clusters. Experiment 3: at most 18 entries differ between clusters.}
	\label{fig: effect of regularization}
\end{figure*}

In order to guarantee a correct cooperation among clusters, we repeat the same experiment as Fig.~\ref{fig: network MSD 6 algorithms} using the adaptive rule in~\eqref{eq: adaptive regularization factors} for adjusting the regularization factors $p_{k\ell}$. The proportionality coefficient in~\eqref{eq: adaptive regularization factors} is set equal to one. As shown in Fig.~\ref{fig: adaptive regularization factors}, when the number of distinct components is small, both $\ell_1$ and reweighted $\ell_1$-norms yield better performance than the diffusion LMS with $\eta=0$. When the number of distinct components increases ($i\in(1000,1500]$), the performance of strategy~\eqref{eq: ATC FBS} with $\ell_1$-norm gets closer to diffusion LMS with $\eta=0$, while the reweighted $\ell_1$-norm still guarantees a gain. Thus, the mechanism proposed in~\eqref{eq: adaptive regularization factors} for the selection of the regularization factors improves the cooperation between nodes belonging to distinct clusters.
\begin{figure}
\centering
\includegraphics[trim = 0mm 0mm 0mm 0mm, clip, scale=0.35]{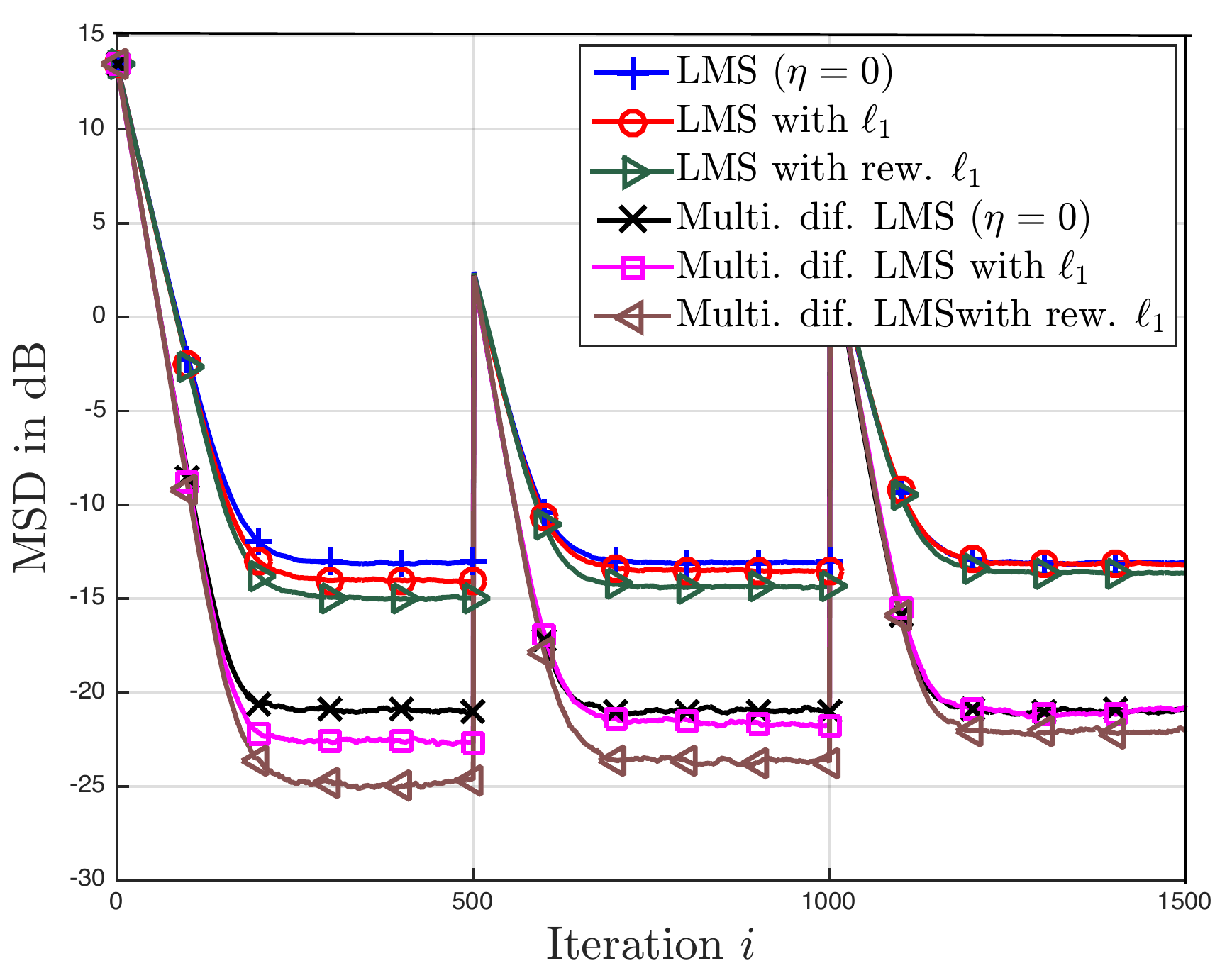}
\caption{Network MSD comparison for the same $6$ different strategies considered in Fig.~\ref{fig: network MSD 6 algorithms} using adaptive regularization factors $p_{k\ell}(i)$.}
\label{fig: adaptive regularization factors}
\end{figure}

Finally, we compare the current multitask diffusion strategy~\eqref{eq: ATC FBS} with two other useful strategies existing in the literature~\cite{plata2014distributed,chen2014multitask}. We consider a stationary environment where the optimum parameter vectors $\{\bw_{\C_j}^o\}_{j=1}^3$ consist of a sub-vector $\boldsymbol{\xi}^o$ of $16$ parameters of global interest to the whole network and a $2\times 1$ sub-vector $\{\boldsymbol{\varsigma}_{\C_j}^o\}$ of common interest to nodes belonging to cluster $\C_j$, namely, $\bw_{\C_j}^o=\col\{\boldsymbol{\xi}^o,\boldsymbol{\varsigma}_{\C_j}^o\}$. The entries of $\boldsymbol{\xi}^o$, $\boldsymbol{\varsigma}_{\C_1}^o$, $\boldsymbol{\varsigma}_{\C_2}^o$, and $\boldsymbol{\varsigma}_{\C_3}^o$ are uniformly sampled from a uniform distribution $\mathcal{U}(-3,3)$. Except for these changes, we consider the same experimental setup described in the first paragraph of the current section. When applying the strategy developed in~\cite{plata2014distributed}, we assume that node $k$ belonging to cluster $\C_j$ is aware that the first $16$ parameters of $\bw_{\C_j}^o$ are of global interest to the whole network while the remaining parameters are of common interest to nodes in cluster $\C_j$. However, the current method~\eqref{eq: ATC FBS} and the algorithm in~\cite{chen2014multitask} do not require such assumption. We run the ATC D-NSPE strategy developed in~\cite{plata2014distributed} using uniform combination weights $a_{\ell k}^w=1/\card\{\N_k\}$ for $\ell\in\N_{k}$ and $a_{\ell k}^{\varsigma_{\C(k)}}=1/\card\{\N_k\cap\C(k)\}$ for $\ell\in\N_{k}\cap\C(k)$, and uniform step-sizes $\mu_k=0.02$ $\forall k$. We run the multitask diffusion strategy developed in~\cite{chen2014multitask} by setting $\{c_{\ell k},a_{\ell k},\rho_{k\ell}\}$ in the same manner described in the first paragraph of the current section, $\mu_k=0.02$ $\forall k$, and $\eta=0.06$. The learning curves of the algorithms are reported in Fig.~\ref{fig: comparison with competitive strategies}. As expected, it can be observed that the cooperation between clusters based on the $\ell_2$-norm~\cite{chen2014multitask} degrades the performance relative to the case of non-cooperative clusters, i.e., $\eta=0$. Indeed, the multitask diffusion strategy developed in~\cite{chen2014multitask} considers squared $\ell_2$-norm co-regularizers to promote the smoothness of the graph signal, whereas, in the current simulation we need to promote the sparsity of the vector $\bw_k^o-\bw_\ell^o$. Furthermore, when the reweighted $\ell_1$-norm is used, our method is able to perform well compared to the strategy developed in~\cite{plata2014distributed} that requires the knowledge of the indices of common and distinct entries in the parameter vectors. We note that recent unsupervised strategies~\cite{chen2016group,plata2015unsupervised} dealing with group of variables rather than variables propose to add a step in order to adapt the cooperation between neighboring nodes based on the group at hand. It is shown in~\cite{chen2016group} that the performance depends heavily on the group decomposition of the parameter vectors.
\begin{figure}
\centering
\includegraphics[trim = 0mm 0mm 0mm 0mm, clip, scale=0.35]{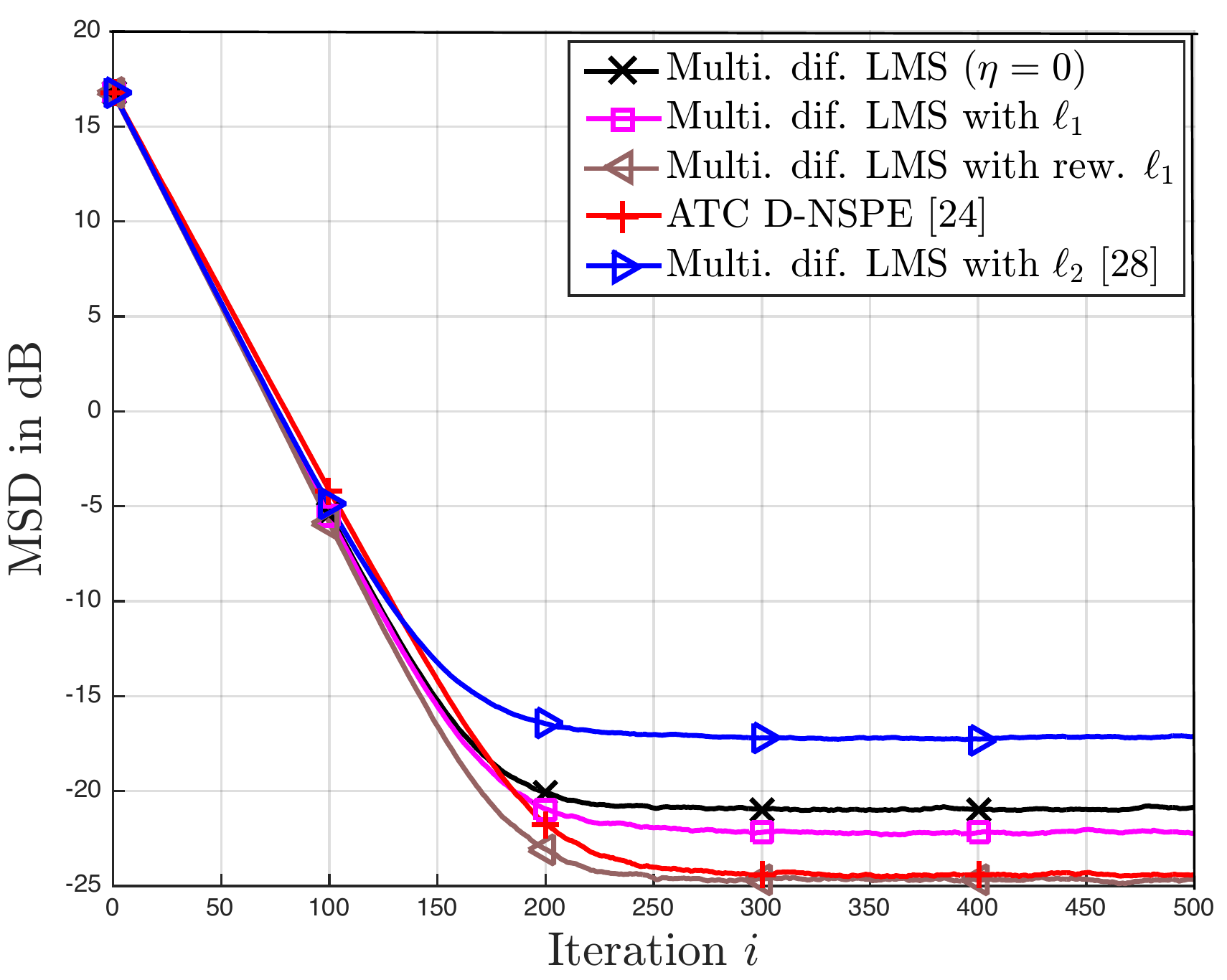}
\caption{Network MSD comparison for $5$ different strategies: standard diffusion without cooperation between clusters (algorithm~\eqref{eq: ATC FBS} with $\eta=0$), our proximal diffusion~\eqref{eq: ATC FBS} with $\ell_1$-norm and reweighted $\ell_1$-norm, the ATC D-NSPE algorithm developed in~\cite{plata2014distributed}, and the multitask diffusion strategy with squared $\ell_2$-norm coregularizers~\cite{chen2014multitask}.}
\label{fig: comparison with competitive strategies}
\end{figure}


\subsection{Distributed spectrum sensing}

Consider a cognitive radio network composed of $N_P$ primary users (PU) and $N_S$ secondary users (SU). To avoid causing harmful interference to the primary users, each secondary user has to detect the frequency bands used by all primary users, even under low signal to noise ratio conditions \cite{lorenzo2013cogradio,sayed2014diffusion,plata2014distributed}. We assume that the secondary users are grouped into $Q$ clusters and that there exists within each cluster a low power interference source (IS). The goal of each secondary user is to estimate the aggregated spectrum transmitted by all active primary users, as well as the spectrum of the interference source present in its cluster.  

In order to facilitate the estimation task of the secondary users, we assume that the power spectrum of the signal transmitted by the primary user $p$ and the interference source $q$ can be represented by a linear combination of $N_B$ basis functions $\phi_m(f)$:
\begin{eqnarray}
	S_p(f)	&=&	\sum_{m=1}^{N_B}\alpha_{pm}\phi_m(f),\quad p=1,\ldots, N_P,\\
	S_q(f)	&=&	\sum_{m=1}^{N_B}\beta_{qm}\phi_m(f),\quad q=1,\ldots, Q,
\end{eqnarray}
where $\alpha_{pm}$, $\beta_{qm}$ are the combination weights, and $f$ is the normalized frequency. Each secondary user $k\in\C_q$ has to estimate the $N_B\times(N_P+1)$ vector $\bUps^o_k=\col\{\balpha^o_1,\ldots,\balpha^o_{N_P},\bbeta_q^o\}$ where $\balpha^o_p=[\alpha_{p1},\ldots,\alpha_{pN_B}]^\top$ and $\bbeta_q^o=[\beta_{q1},\ldots,\beta_{qN_B}]^\top$. Let $\ell_{p,k}(i)$ denote the path loss factor between the primary user $p$ and the secondary user $k$ at time $i$. Let also $\ell'_{q,k}(i)$ denote the path loss factor between the interference source $q$ and the secondary user $k$ at time $i$. Then, the power spectrum sensed by node $k\in\C_q$ at time $i$ and frequency $f_j$ can be expressed as follows:
\begin{equation}
\label{eq: relation 1}
r_{k,j}(i)=\sum_{p=1}^{N_P}\ell_{p,k}(i)S_p(f_j)+\ell'_{q,k}(i)S_q(f_j)+z_{k,j}(i),
\end{equation}
where $z_{k,j}(i)$ is the sampling noise at the $j$-th frequency assumed to be zero-mean Gaussian with variance $\sigma^2_{z_{k,j}}$. At each time instant $i$, node $k$ observes the power spectrum over $N_F$ frequency samples. Let $\br_k(i)$ and $\bz_k(i)$ be the $N_F\times 1$ vectors whose $j$-th entries are $r_{k,j}(i)$ and $z_{k,j}(i)$, respectively. Using~\eqref{eq: relation 1}, we can establish the following linear data model for node $k\in\C_q$:
\begin{equation}
\br_k(i)=\bPhi_k(i)\bUps^o_k+\bz_k(i),
\end{equation}
where $\bPhi_k(i)\triangleq[\ell_{1,k}(i),\ldots,\ell_{N_P,k}(i),\ell'_{q,k}(i)]\otimes\bPhi$ with $\bPhi$ the $N_F\times N_B$ matrix whose $j$-th row contains the magnitudes of the $N_B$ basis functions at the frequency sample $f_j$. 

To show the effect of multitask learning with sparsity-based regularization, we consider a cognitive radio network consisting of $N_P=2$ primary users and $N_S=13$ secondary users decomposed into $4$ clusters as shown in Fig.~\ref{fig: cognitive topology}. The power spectrum is represented by a combination of $N_B=20$ Gaussian basis functions centered at the normalized frequency $f_m$ with variance $\sigma^2_m=0.001$ for all $m$:
\begin{equation}
\phi_m(f)=\exp^{-\frac{(f-f_m)^2}{2\sigma^2_m}},
\end{equation} 
where the central frequencies $f_m$ are uniformly distributed. The combination vectors are set to:
\begin{equation}
	\begin{split}
		\bUps_{\C_1}^o		&=[\cb{0}_{1 \times 4} ~ 1 ~ 1 ~ \cb{0}_{1 \times 14} ,\cb{0}_{1 \times 14} ~ 1~ 1 ~ \cb{0}_{1 \times 4},0~ 0.3 ~0.3~\cb{0}_{1 \times 17}]^\top	\\
		\bUps_{\C_2}^o		&=[\cb{0}_{1 \times 4} ~ 1 ~ 1 ~ \cb{0}_{1 \times 14} ,\cb{0}_{1 \times 14} ~ 1~ 1 ~ \cb{0}_{1 \times 4},\cb{0}_{1 \times 20}]^\top	\\
		\bUps_{\C_3}^o		&=[\cb{0}_{1 \times 4} ~ 1 ~ 1 ~ \cb{0}_{1 \times 14} ,\cb{0}_{1 \times 14} ~ 1~ 1 ~ \cb{0}_{1 \times 4},0~ 0.3 ~\cb{0}_{1 \times 16} ~ 0.3 ~ 0]^\top	\\
		\bUps_{\C_4}^o		&=[\cb{0}_{1 \times 4} ~ 1 ~ 1 ~ \cb{0}_{1 \times 14} ,\cb{0}_{1 \times 14} ~ 1~ 1 ~ \cb{0}_{1 \times 4},\cb{0}_{1 \times 17} ~ 0.3 ~ 0.3 ~ 0]^\top.			
	\end{split}
\end{equation}
\begin{figure}
\centering
\includegraphics[trim = 0mm 0mm 0mm 0mm, clip, scale=0.33]{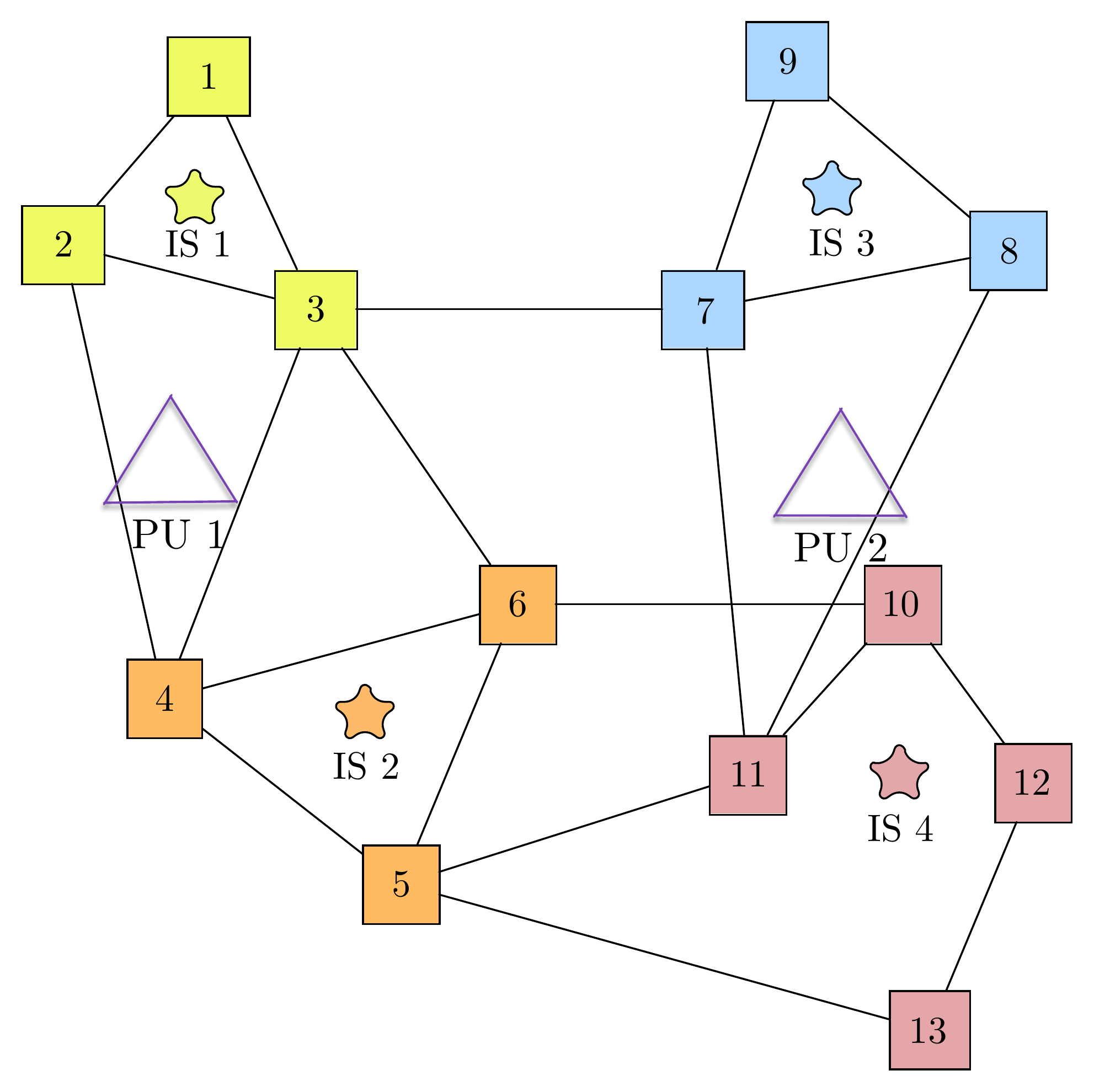}
\caption{A cognitive radio network consisting of $2$ primary users and $13$ secondary users grouped into $4$ clusters containing each an interference source IS.}
\label{fig: cognitive topology}
\end{figure}
We consider $N_F=80$ frequency samples. Based on the free propagation theory, we set the deterministic path loss factor $\overline{\ell}_{p,k}$ to the inverse of the squared distance between the transmitter $p$ and the receiver $k$. At time instant $i$, we set $\ell_{p,k}(i)= \overline{\ell}_{p,k}+\delta\ell_{p,k}(i)$ with $\delta\ell_{p,k}(i)$ a zero-mean random Gaussian variable with standard deviation $0.1\overline{\ell}_{p,k}$. The secondary user $k$ estimates $\ell_{p,k}(i)$ according to the following model:
\begin{equation}
	\hat{\ell}_{p,k}(i)=\left\lbrace
	\begin{array}{ll}
		\overline{\ell}_{p,k},	\quad&\text{if}\quad \ell_{p,k}(i)>\ell_0,\\
		0,				\quad&\text{otherwise}
	\end{array}
	\right.
\end{equation}
with $\ell_0$ a threshold value. The same rule is used to set the path loss factor between the interference sources and the secondary users. We run the ATC diffusion algorithm~\eqref{eq: ATC FBS} with the following adaptation step:
\begin{equation}
\label{eq: adaptation step cogradio}
\bpsi_k(i+1)=\bUps_k(i)+\mu_k\hat{\bPhi}_{k}^\top(i)[\br_k(i)-\hat{\bPhi}_{k}(i)\bUps_k(i)],\\
\end{equation}
with $\bUps_k(i)$ the estimate of $\bUps_k^o$ at time instant $i$. The sampling noise $z_{k\ell,j}(i)$ is assumed to be a zero-mean random Gaussian variable with standard deviation $0.01$. The combination coefficients $\{a_{\ell k}\}$ and regularization factors $\{\rho_{k\ell}\}$ are set in the same way as in the previous experimentation. 

The MSD learning curves are averaged over $50$ Monte-Carlo runs. We run the multitask diffusion LMS~\eqref{eq: ATC FBS} in two different situations. In the first scenario, we do not allow any cooperation between clusters by setting $\eta=0$. In the second scenario, we set the regularization strength $\eta$ to $0.01$ and we use the $\ell_1$-norm as co-regularizing function. As can be seen in Fig.~\ref{fig: cognitive MSD}, the network MSD performance is significantly improved by cooperation among clusters. For comparison purposes, we also run the ATC D-NSPE strategy developed in~\cite{plata2014distributed} and the multitask diffusion strategy with $\ell_2$-norm developed in~\cite{chen2014multitask}. For the ATC D-NSPE strategy we assume that nodes are aware that the first $N_P\times N_B$ components of the vector $\bUps^o_k$ are of global interest to the whole network and that the remaining components are of common interest to the cluster $\C(k)$. The link weights $\{a_{\ell k}, c_{\ell k}, \rho_{k\ell}, a_{\ell k}^w, a_{\ell k}^{\varsigma_{\C(k)}}\}$ are set in the same manner as the experiment in Fig.~\ref{fig: comparison with competitive strategies}. It can be observed from Fig.~\ref{fig: cognitive MSD} that our strategy performs well without the need to know the parameters of global interest and the parameters of common interest during the learning process. Figure~\ref{fig: cognitive PSD} shows the estimated power spectrum density for nodes $2$, $4$, $7$, and $13$ when running the multitask diffusion strategy~\eqref{eq: ATC FBS} with $\eta=0$ (left) and $\eta=0.01$ (right). In the left plot, we observe that the clusters are able to estimate their interference source. However, depending on the distance to the primary users, the secondary users do not always succeed in estimating the power spectrum transmitted by all active primary users. For example, clusters $1$ and $2$ are not able to estimate the power spectrum transmitted by PU$2$. As shown in the right plot, regardless of the distance between primary and secondary users, each secondary user is able to estimate the aggregated power spectrum transmitted by all the primary users and its own interference source by cooperating with nodes belonging to neighboring clusters.
\begin{figure}
\centering
\includegraphics[trim = 0mm 0mm 0mm 0mm, clip, scale=0.36]{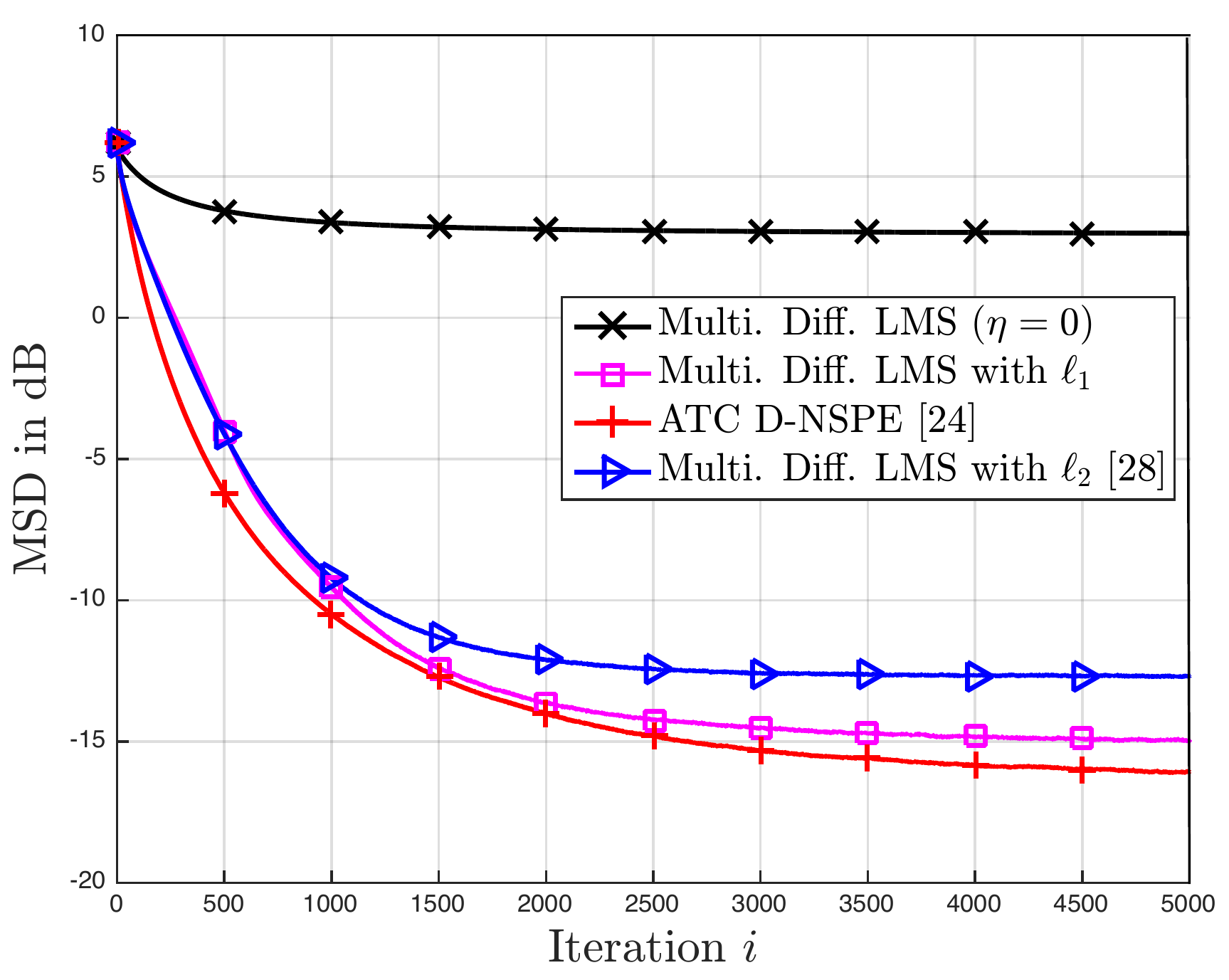}
\caption{Network MSD comparison for $4$ different algorithms: standard diffusion LMS without cooperation between clusters ($\eta=0$), our proximal diffusion~\eqref{eq: ATC FBS} with $\ell_1$-norm regularizer, the ATC D-NSPE algorithm developed in~\cite{plata2014distributed}, and the multitask diffusion strategy~\cite{chen2014multitask}.}
\label{fig: cognitive MSD}
\end{figure}
\begin{figure*}
\centering
\includegraphics[trim = 10mm 70mm 20mm 70mm, clip, scale=0.36]{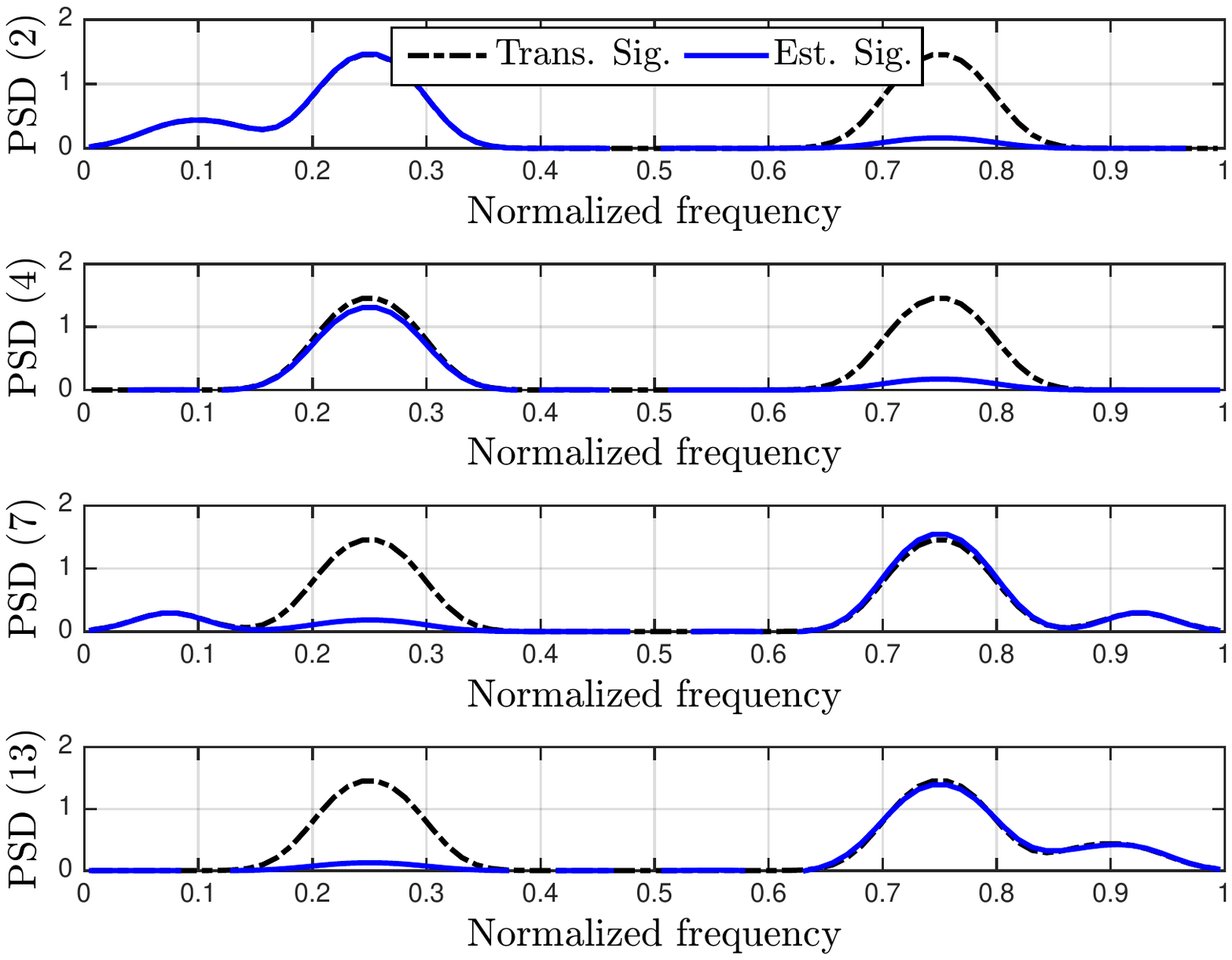}
\includegraphics[trim = 10mm 70mm 20mm 70mm, clip, scale=0.36]{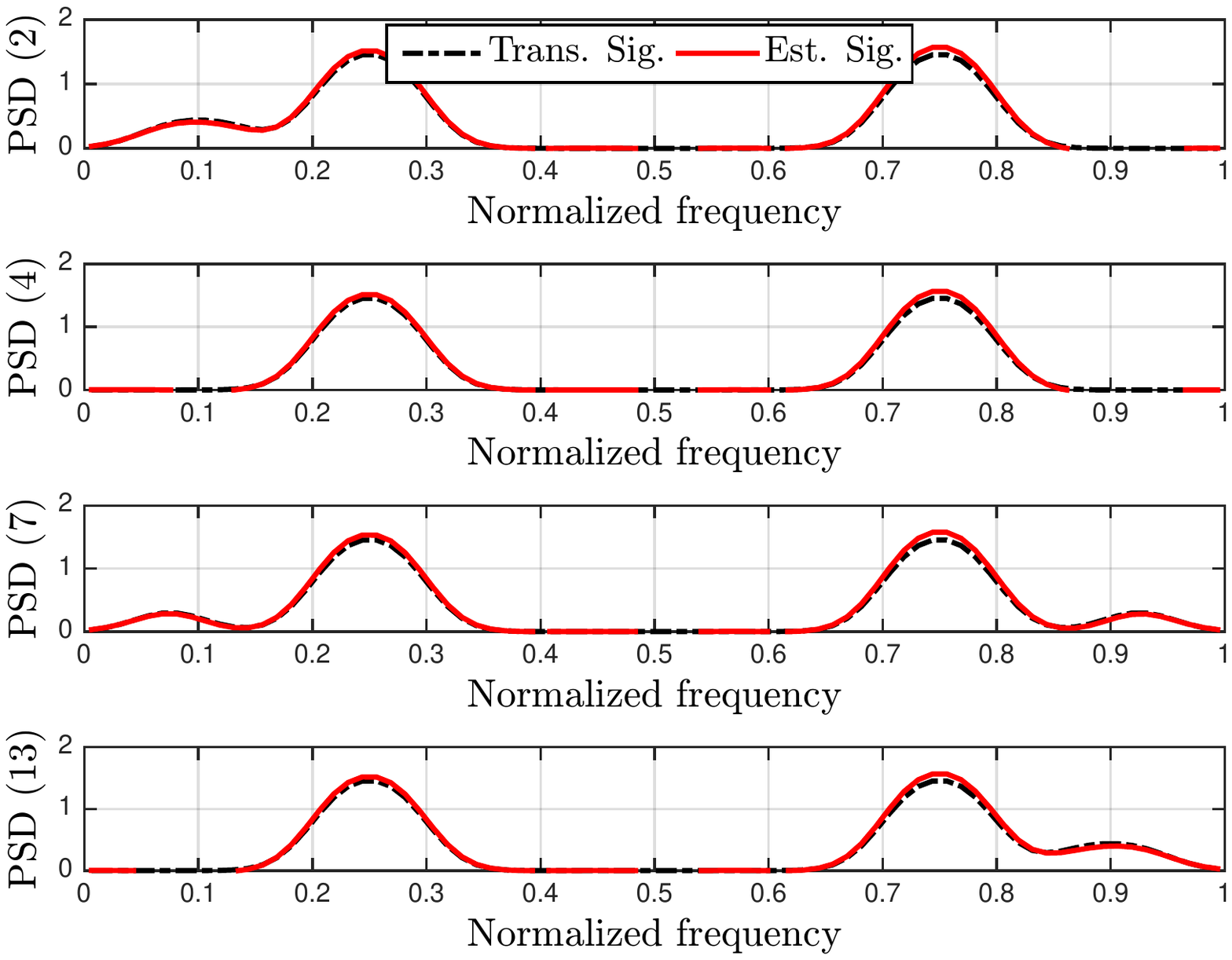}
\caption{PSD estimation for nodes 2 ($\C_1$), 4 ($\C_2$), 7 ($\C_3$), and 13 ($\C_4$). Left: noncooperating clusters (multitask strategy~\eqref{eq: ATC FBS} with $\eta=0$). Right: cooperating clusters (multitask strategy~\eqref{eq: ATC FBS} with $\eta\neq0$).}
\label{fig: cognitive PSD}
\end{figure*}


\section{Conclusion and perspectives}
\label{sec: conclusion}

In this work, we considered multitask learning problems over networks where the optimum parameter vectors to be estimated by neighboring clusters have a large number of similar entries and a relatively small number of distinct entries. It then becomes advantageous to develop distributed strategies that involve cooperation among adjacent clusters in order to exploit these similarities. A diffusion forward-backward splitting algorithm with $\ell_1$-norm and reweighed $\ell_1$-norm co-regularizers was derived to address this problem. A closed-form expression for the proximal operator was derived to achieve higher efficiency. Conditions on the step-sizes to ensure convergence of the algorithm in the mean and mean-square sense were derived. Finally, simulation results were presented to illustrate the benefit of cooperating to promote similarities between estimates. Future research efforts will be focused on exploiting other sparsity promoting co-regularizers. Perspectives also include the derivation of other forms of cooperation depending on prior information.


\bibliographystyle{IEEEbib}
\bibliography{Proxref}

\begin{IEEEbiography}[{\includegraphics[width=1in,height=1.25in,clip,keepaspectratio]{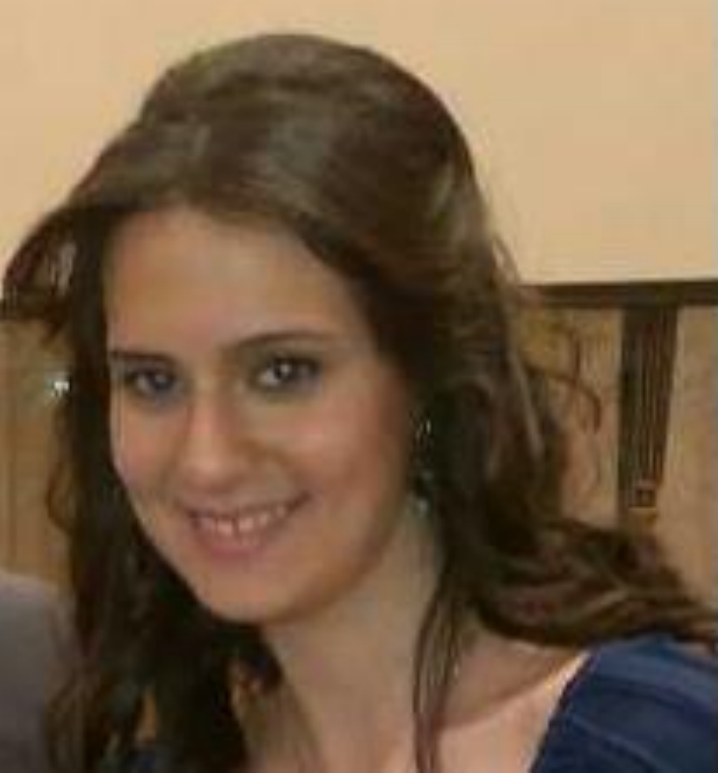}}
]{Roula Nassif}

was born in Beirut, Lebanon. She received the bachelor's degree in Electrical Engineering from the Lebanese University, Lebanon, in 2013. She received the M.S. degrees in Industrial Control and Intelligent Systems for Transport from the Lebanese University, Lebanon, and from Compi\`egne University of Technology, France, in 2013. Since October 2013 she is a Ph.D. student at the Lagrange Laboratory (University of Nice Sophia Antipolis, CNRS, Observatoire de la C\^ote d'Azur). Her research activity is focused on distributed optimization over multitask networks. 

\end{IEEEbiography}

\begin{IEEEbiography}[{\includegraphics[width=1in,height=1.25in,trim = 50mm 50mm 50mm 50mm, clip,keepaspectratio]{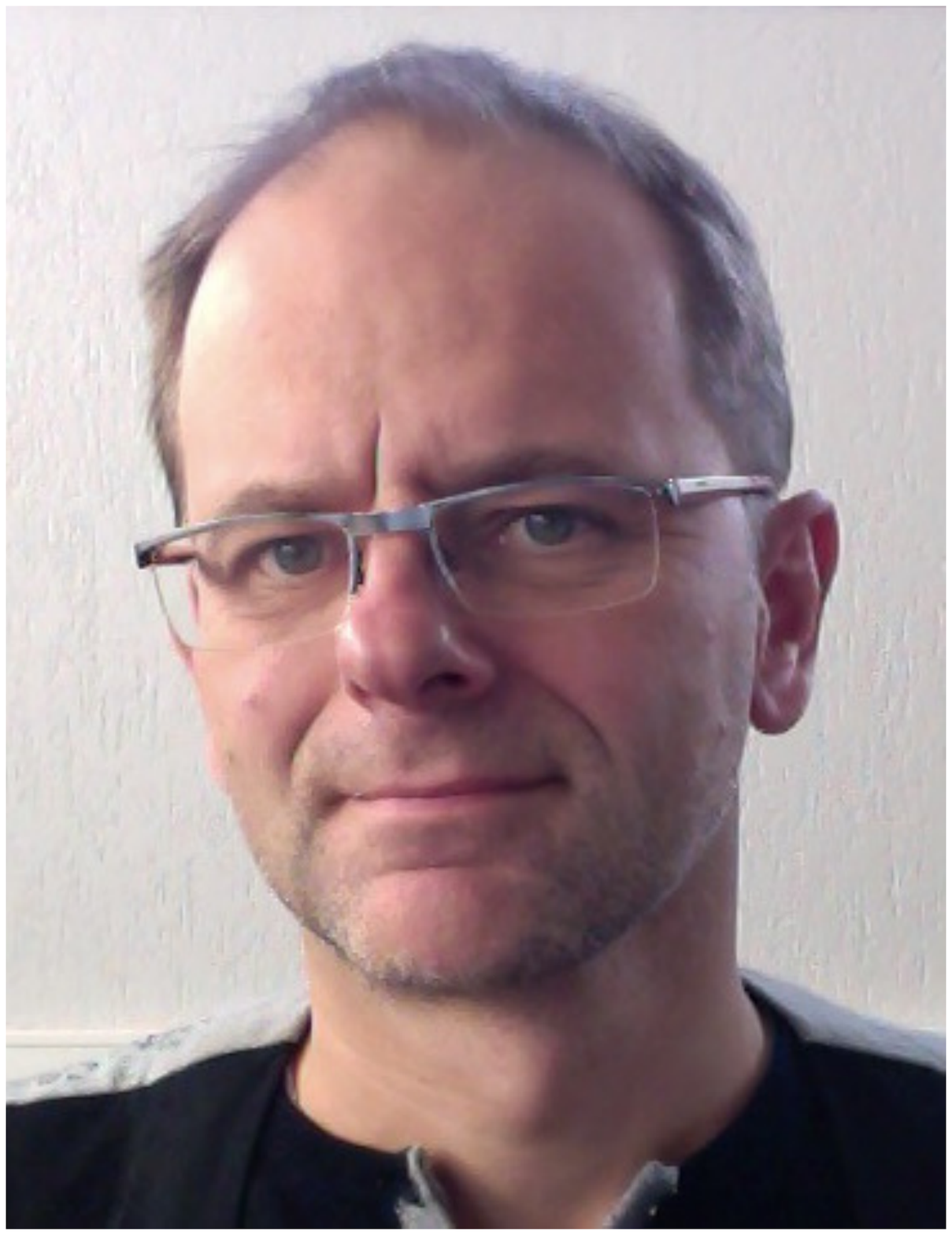}}]{C{\'e}dric Richard}
(S'98--M'01--SM'07)  received the Dipl.-Ing. and the M.S. degrees in 1994, and the Ph.D. degree in 1998, from Compi\`egne University of Technology, France, all in electrical and computer engineering. He is a Full Professor at the Universit\'e C\^ote d'Azur, France. He was a junior member of the Institut Universitaire de France in 2010-2015.

His current research interests include statistical signal processing and machine learning. C\'edric Richard is the author of over 250 papers. He was the General Co-Chair of the IEEE SSP Workshop that was held in Nice, France, in 2011. He was the Technical Co-Chair of EUSIPCO 2015 that was held in Nice, France, and of the IEEE CAMSAP Workshop 2015 that was held in Cancun, Mexico. He serves as a Senior Area Editor of the IEEE Transactions on Signal Processing and as an Associate Editor of the IEEE Transactions on Signal and Information Processing over Networks since 2015. He is also an Associate Editor of Signal Processing Elsevier since 2009. C\'edric Richard is member of the Machine Learning for Signal Processing (MLSP TC) Technical Committee, and served as member of the Signal Processing Theory and Methods (SPTM TC) Technical Committee in 2009-2014.
\end{IEEEbiography}

\begin{IEEEbiography}[{\includegraphics[width=1in,height=1.25in,clip,keepaspectratio]{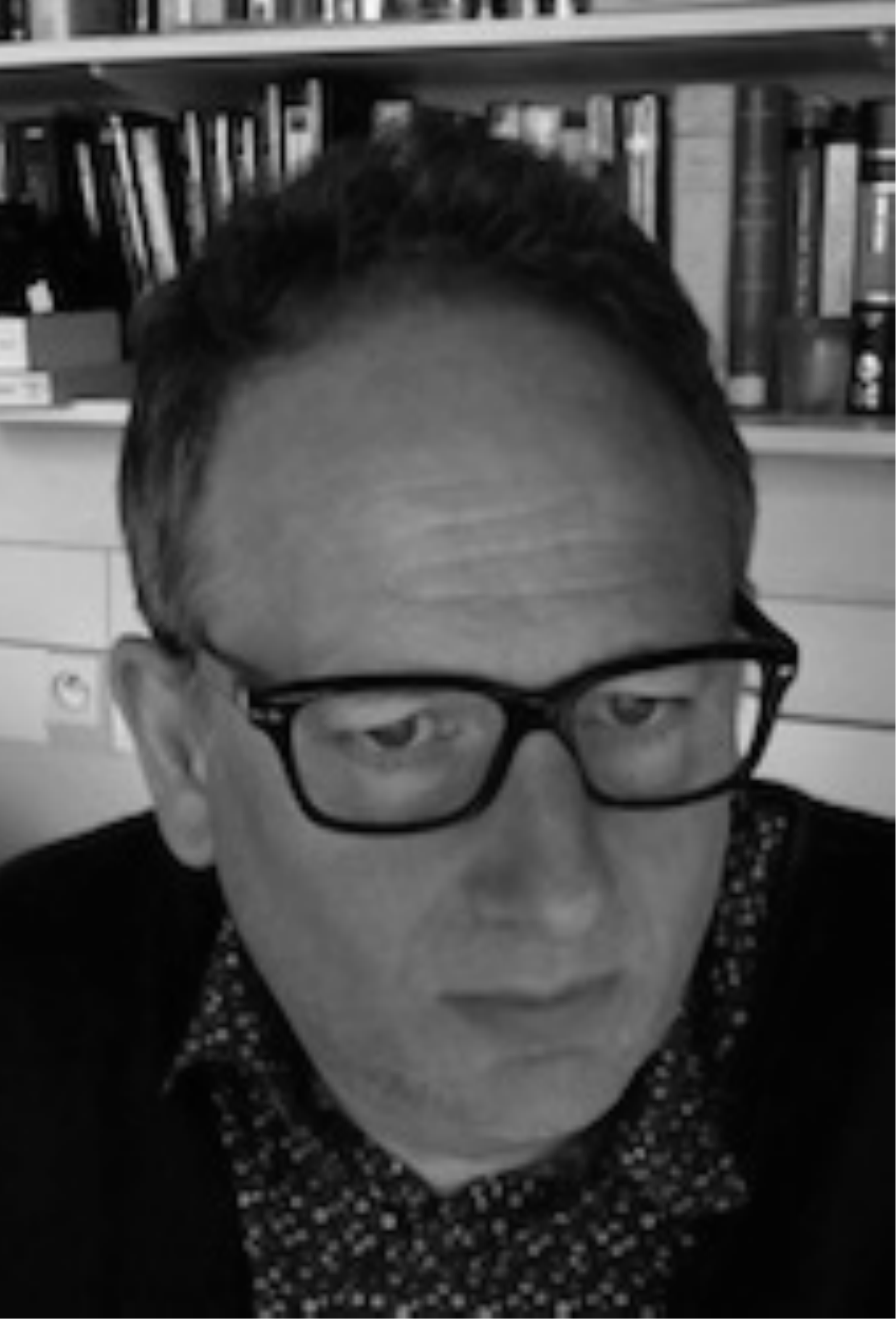}}
]{Andr{\'e} Ferrari}

(SM'91-M'93) received the Ing\'enieur degree from \'Ecole Centrale de Lyon, Lyon, France, in 1988 and the M.Sc. and Ph.D. degrees from the University of Nice Sophia Antipolis (UNS), France, in 1989 and 1992, respectively, all in electrical and computer engineering. 

He is currently a Professor at UNS. He is a member of the Joseph-Louis Lagrange Laboratory (CNRS, OCA), where his research activity is centered around statistical signal processing and modeling, with a particular interest in applications to astrophysics.

\end{IEEEbiography}

\begin{IEEEbiography}[{\includegraphics[width=1in,height=1.25in,clip,keepaspectratio]{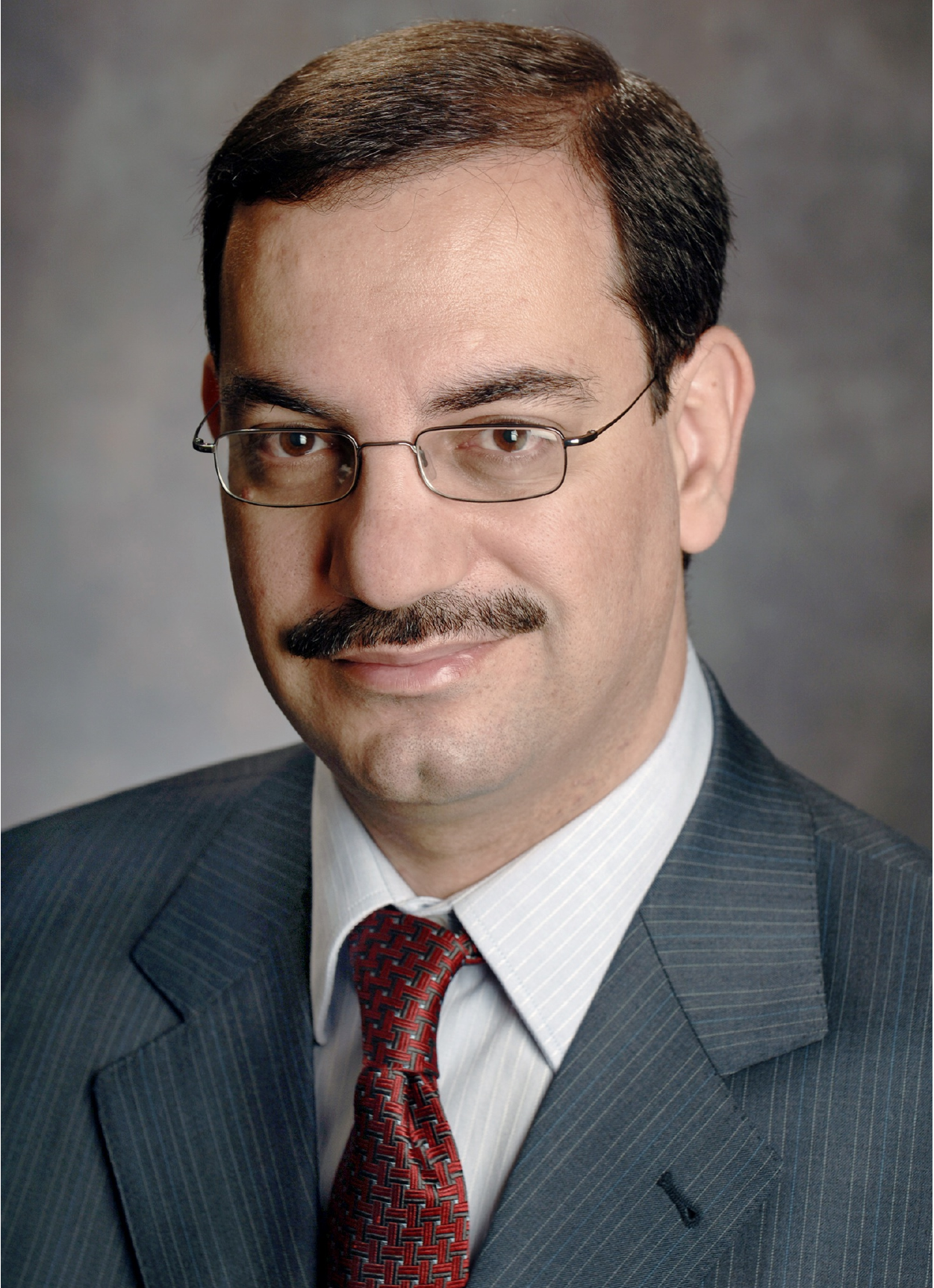}}
]{Ali H. Sayed}

(S'90-M'92-SM'99-F'01) is a Distinguished Professor and past Chairman of electrical engineering at UCLA where he directs the UCLA Adaptive Systems Laboratory (www.ee.ucla.edu/asl). An author or co-author of over 480 scholarly publications and six books, his research involves several areas including adaptation and learning, system theory, statistical signal processing, network science, and information processing theories. His work has been recognized with several awards including the 2015 Education Award from the IEEE Signal Processing Society, the 2014 Papoulis Award from the European Association for Signal Processing, the 2013 Meritorious Service Award and the 2012 Technical Achievement Award from the IEEE Signal Processing Society, the 2005 Terman Award from the American Society for Engineering Education, the 2003 Kuwait Prize, and the 1996 IEEE Donald G. Fink Prize. He has been awarded several Best Paper Awards from 
the IEEE (2002, 2005, 2012, 2014) and EURASIP (2015) and is a Fellow of both IEEE and the American Association for the Advancement of Science (AAAS). He is recognized as a Highly Cited Researcher by Thomson Reuters. He is President-Elect of the IEEE Signal Processing Society (2016-2017).
\end{IEEEbiography}
\end{document}